\documentclass[12pt]{article}

\usepackage[T1]{fontenc}
\usepackage[utf8]{inputenc}
\usepackage{lmodern}
\usepackage{microtype}

\usepackage{amsmath, amssymb, amsthm}
\usepackage{bm}
\usepackage{mathtools}

\usepackage{booktabs}
\usepackage{array}
\usepackage{tabularx}
\usepackage{multirow}
\usepackage{graphicx}
\graphicspath{{./}}
\usepackage{tikz}
\usetikzlibrary{arrows.meta, positioning}
\usepackage{float}
\usepackage{placeins}
\usepackage[font=small, labelfont=bf]{caption}

\usepackage{algorithm}
\usepackage{algpseudocode}

\usepackage{enumitem}
\usepackage{parskip}
\usepackage{natbib}
\usepackage[colorlinks=true, allcolors=black]{hyperref}

\newtheorem{theorem}{Theorem}
\newtheorem{proposition}[theorem]{Proposition}

\newtheorem{corollary}[theorem]{Corollary}

\newtheorem{assumption}{Assumption}
\newtheorem{remark}{Remark}

\newcommand{\cC}{\mathcal{C}}
\newcommand{\cF}{\mathcal{F}}
\newcommand{\cO}{\mathcal{O}}
\newcommand{\cA}{\mathcal{A}}
\newcommand{\cN}{\mathcal{N}}

\newcommand{\cG}{\mathcal{G}}

\newcommand{\bbE}{\mathbb{E}}
\newcommand{\bbP}{\mathbb{P}}
\DeclareMathOperator{\sign}{sign}
\DeclareMathOperator{\Var}{Var}

\DeclareMathOperator{\rank}{rank}

\newcommand{\Sthree}{\mathrm{S}^3}
\newcommand{\bSthree}{\mathbf{S}^{3}}

\newcommand{\anon}{1}

\def\spacingset#1{\renewcommand{\baselinestretch}%
{#1}\small\normalsize}

\begin{document}

\spacingset{1}

\if1\anon
{
  \title{\bf Sequential Supersaturated Screening Experiments}
  \author{%
    Songqiao Han$^{a,*}$, Kalliopi Mylona$^{a}$,
    Steven Gilmour$^{a}$, and Matteo Borrotti$^{b}$
    \\[6pt]
    \normalsize $^{a}$Department of Mathematics, King's College London
    \\
    \normalsize $^{b}$Department of Economics, Management and Statistics,
    University of Milano-Bicocca
    \\[4pt]
    \normalsize $^{*}$Corresponding author:
    \texttt{songqiao.han@kcl.ac.uk}}
  \date{}
  \maketitle
} \fi

\if0\anon
{
  \bigskip
  \bigskip
  \bigskip
  \begin{center}
    {\LARGE\bf Sequential Supersaturated Screening Experiments}
  \end{center}
  \medskip
} \fi

\bigskip
\begin{abstract}
Supersaturated screening experiments study many candidate factors with
few runs. The experimenter must first identify the active factors,
then optimize the response. One-shot regularized regression tends to
select many factors, and the second-order response surface model in the
selected factors is then large and needs many further runs to fit. We
propose \textbf{S}equential \textbf{S}upersaturated
\textbf{S}creening ($\bSthree$), a two-stage framework for supersaturated
designs. $\Sthree$ builds each screening round by coordinate exchange
under a new positive-cone design criterion that uses no
practitioner-chosen Welch calibration constant. It then removes
low-importance factors one round at a time, using a graduated quantile
rule. The criterion scores a design by how far its column correlations sit
above the Welch lower bound, and it adapts to the current number of runs
and candidate factors. We give a round-by-round
bound on the probability that a noise factor is ever fixed, under
conditions on noise survival and commitment at each round, together
with an explicit upper bound on the total number of Stage 1 runs. Across simulated screening problems and the Borehole benchmark,
$\Sthree$ achieves lower Type I error and higher $F_1$ scores than
one-shot cross-validated Lasso (LassoCV). It also improves
optimization quality in most settings and runs faster.
\end{abstract}

\noindent%
\vfill

\newpage
\spacingset{1.8}

\section{Introduction}
\label{sec:intro}

Experimentation is essential for understanding, improving, and optimizing industrial and scientific processes. The response of such a process can depend on a large number of input variables. When experimental runs are expensive, designs that require few runs are particularly valuable. 
A screening experiment varies a set of input variables, from here on the experimental factors, over a small number of runs to identify the active factors: those that make a non-negligible contribution to the response through a linear main effect, a quadratic main effect, or an interaction. With $p$ candidate factors and fewer runs than the $p + 1$ coefficients of a model with an intercept and linear main effects, the screening design is supersaturated, and the main effects cannot all be estimated from it. After screening, the experimenter still has to optimize the response, which requires a second-order response surface model in the factors the screen selects.

A natural strategy therefore runs in two stages. Stage 1 screens the $p$ candidates down to a shortlist, and Stage 2 fits and optimizes a second-order response surface model in the shortlisted factors. Both stages consume experimental runs, and the cost of Stage 2 is set by the screening output. The number of coefficients of a full second-order model grows quadratically with the number of selected factors, so every additional factor retained at Stage 1 increases the number of runs Stage 2 needs.

Existing sparse-regression screeners address a different first-stage problem. Given one fixed set of $n$ screening runs and $p$ candidate
factors, they choose a sparse model from that batch, usually for prediction or variable selection, rather than to keep the follow-up response surface model small. Cross-validated Lasso (LassoCV) \citep{tibshirani1996regression} is a standard one-shot baseline, but its penalty path is chosen by predictive cross-validation and takes no account of the second-order model that must be fitted next. In the two-level settings studied here, this default tuning selects comparatively many factors, so the second-order model in the selected factors is large and typically needs many further runs. The same issue arises for the ElasticNet \citep{zou2005regularization}, Adaptive Lasso \citep{zou2006adaptive}, SCAD \citep{fan2001variable}, and SIS \citep{fan2008sure}.

The design-of-experiments community has developed a growing set of tools for constructing and analyzing supersaturated designs (SSDs). 
Constrained $\Var(s_+)$ designs \citep{weese2021strategies} target column-pair correlations. Pareto-optimal SSD selection \citep{singh2023selection} and GDS-ARM aggregation over random interaction models \citep{singh2024factor} target one-shot selection or analysis criteria under their stated models. All of these methods operate in a one-shot paradigm: they analyze a single batch of data and do not adapt the design as evidence accumulates.

Sequential approaches to supersaturated screening exist, but to our knowledge none provides an end-to-end adaptive pipeline at the scale of this paper. The Stepwise Response Refinement Screener of \citet{phoa2014stepwise} iteratively refines the response after each identified factor on a fixed supersaturated design, which is a sequential analysis rather than a sequential experimental design. \citet{gutman2014augmenting} augment an existing SSD with a single follow-up set of runs chosen via Bayesian $D$-optimality, which is a two-stage rather than multistage procedure. Neither integrates a multi-round elimination rule with a supersaturated two-level design construction.

A separate line of work analyzes the geometric structure of the inner product matrix of an SSD. 
\citet[Theorem 2]{stallrich2025optimal} shows that, under their lasso sign-recovery probability criterion and when the active signs are known, a structure in which all column pairs share a positive correlation is ideal for two-level SSDs. Such exact symmetry need not be attainable for every $(n,p)$ design size. We therefore use it as a qualitative target and calibrate the column-pair inner products against the \citet{welch1974lower} lower bound, which gives the scale of unavoidable aliasing among $p > n$ columns. 

To close this gap we propose \textbf{S}equential
\textbf{S}upersaturated \textbf{S}creening ($\bSthree$), an adaptive method that screens and then optimizes. 
Instead of committing all runs at once, $\Sthree$ collects them in several small rounds and drops low-importance candidate factors round by round, stopping at a short list whose second-order response surface model is small enough to fit and optimize with a modest number of further runs. The experimenter does not divide the runs between the two stages in advance. Our main contributions are:
\begin{itemize}
  \item \textbf{An adaptive screening-to-optimization framework for
  supersaturated screening.}
  $\Sthree$ integrates multi-round elimination of candidate factors with
  coordinate-exchange design construction, complementing one-shot \citep{weese2021strategies,singh2023selection,singh2024factor}, fixed-design sequential-analysis
  \citep{phoa2014stepwise}, and single follow-up augmentation
  \citep{gutman2014augmenting} approaches.
  \item \textbf{The $\Var(s_+^{\mathrm{W}})$ design criterion}, with no practitioner-chosen Welch calibration constant: its weight is determined by the run and factor counts via the \citet{welch1974lower} lower
  bound. The criterion drives all column-pair correlations toward the shared positive sign structure of
  \citet[Theorem 2]{stallrich2025optimal}, at the magnitude scale supersaturation imposes (Section \ref{sec:welch}).
  \item \textbf{Theoretical support for the two-stage design.} An
  explicit upper bound on the total number of Stage 1 runs,
  $N_{\mathrm{S1}} \le p / q + O(\log p)$, and a round-by-round bound on the probability that an inactive factor is committed into Stage 2, with assumptions discussed in Section \ref{sec:theory} and round-level diagnostics in Appendix \ref{app:noise_diagnostics}.
  \item \textbf{Empirical comparison across simulation and a benchmark application.}
  Across simulated cases and the Borehole benchmark function \citep{harper1983sensitivity}, $\Sthree$ has lower Type I error and higher $F_1$ scores (defined in Appendix \ref{app:simulation_details}) than one-shot LassoCV. $\Sthree$ also improves optimization quality in most settings and runs faster (Sections \ref{sec:sim} and \ref{sec:borehole}).
\end{itemize}

\section{The \texorpdfstring{$\Sthree$}{S3} Framework}
\label{sec:method}

$\Sthree$ has two stages: Stage 1 screens the $p$ candidate factors down to a small active set, and Stage 2 fits and optimizes a second-order response surface over it. We first fix the model and notation, then present the design criterion used to construct each round of runs, and finally describe the two stages.

\subsection{Problem Setup and Notation}
\label{sec:feasibility}

Let $p$ denote the total number of candidate factors.
Each factor is coded linearly to $[-1,1]$, giving
$\mathbf{x}=(x_1,\ldots,x_p)^\top\in[-1,1]^p$. 
Stage 1 uses two-level designs on the current candidate factors.
These runs are constructed by coordinate exchange.
Stage 2 adds runs at levels $-1,0,+1$ to fit a second-order response surface, then optimizes it over $[-1,1]$ for each retained factor.
The response is modeled as a second-order response surface with measurement noise and a batch effect. A batch is a group of runs that shares a common nuisance shift, such as a day effect. The observation model is
\begin{equation}
  y_i = f(\mathbf{x}_i) + b_{d(i)} + \varepsilon_i,
  \qquad
  \varepsilon_i \overset{\text{i.i.d.}}{\sim} \mathcal{N}(0, \sigma^2),
  \label{eq:obs}
\end{equation}
where
\begin{equation}
  f(\mathbf{x})
  = \beta_0
  + \sum_{j \in \cA_{\mathrm{m}}} \beta_j x_j
  + \sum_{(j,\ell) \in \cA_{\mathrm{int}}} \beta_{j\ell}\, x_j x_\ell
  + \sum_{j \in \cA_{\mathrm{q}}} \beta_{jj}\, x_j^2,
  \label{eq:model}
\end{equation}
$d(i)$ is the batch index for observation $i$, and $b_{d(i)} \overset{\text{i.i.d.}}{\sim} \mathcal{N}(0, \sigma_b^2)$ is a
batch-level mean shift independent of $\mathbf{x}$. 
The sets $\cA_{\mathrm{m}}$, $\cA_{\mathrm{int}}$, and $\cA_{\mathrm{q}}$ index nonzero linear main effects, nonzero two-factor interactions, and nonzero quadratic terms. The corresponding factor-level active set is
$\cA = \cA_{\mathrm{m}} \cup \{j : (j,\ell) \in \cA_{\mathrm{int}}\}
      \cup \{\ell : (j,\ell) \in \cA_{\mathrm{int}}\}
      \cup \cA_{\mathrm{q}}$,
with $|\cA| = k_{\text{true}}$.
Let
\begin{equation}
  \mathbf{x}^\star
  = \arg\max_{\mathbf{x} \in [-1, +1]^p} f(\mathbf{x})
  \label{eq:xstar}
\end{equation}
denote the true optimum. The practical objective is to return an estimated setting $\hat{\mathbf{x}}$ with small regret relative to $\mathbf{x}^\star$, using few runs in total. Writing $N_{\mathrm{S1}}$ and $N_{\mathrm{S2}}$ for the numbers of Stage 1 and Stage 2 runs, the total $N = N_{\mathrm{S1}} + N_{\mathrm{S2}}$ is an outcome of the adaptive procedure, set by how many screening rounds the data require, not a figure fixed in advance. Stage 1 screens at the $\pm 1$ boundary, and Stage 2 augments and optimizes the second-order surface over the continuous box $[-1, +1]^{k_{\mathrm{sel}}}$ on the selected coordinates, where $k_{\mathrm{sel}}$ is the number of factors Stage 1 selects. We measure optimization quality by the optimality loss
\begin{equation}
  \mathrm{LOSS}
  = \max\!\bigl(0,\; f(\mathbf{x}^\star) - f(\hat{\mathbf{x}})\bigr).
  \label{eq:loss}
\end{equation}
In exact arithmetic the quantity inside the maximum is nonnegative. The outer maximum prevents small numerical overshoots from being reported as negative loss.

\subsection{The \texorpdfstring{$\Var(s_+^{\mathrm{W}})$}{Var(s+W)}
  Criterion}
\label{sec:welch}

Stage 1 collects its runs in small rounds, and each round's new runs are constructed under a design criterion, which we now present. The design quality at each round is governed by the column-pair inner products of the cumulative two-level design matrix. Let $X \in \{-1,+1\}^{n
\times p}$ denote the cumulative design, and let $\mathbf{x}_j \in \{-1,+1\}^n$ denote its $j$-th column (the $n$-vector of design levels for factor $j$). Define the column-pair inner products
\begin{equation}
  s_{ij} \;=\; \mathbf{x}_i^{\!\top}\mathbf{x}_j
        \;=\; \sum_{r=1}^{n} X_{r,i}\,X_{r,j},
  \qquad i, j \in \{1,\ldots,p\},\ i \ne j,
  \label{eq:sij}
\end{equation}
and write $M = \binom{p}{2}$ for the number of unordered column pairs. Throughout we use the explicit notation
\[
  \overline{s^2} \;\equiv\; \frac{1}{M}\sum_{i<j} s_{ij}^2,
  \qquad
  (\bar s)^2 \;\equiv\; \Bigl(\frac{1}{M}\sum_{i<j} s_{ij}\Bigr)^{\!2},
\]
where $\bar s = M^{-1} \sum_{i<j} s_{ij}$ is the average off-diagonal inner product, so that $\Var(s) = \overline{s^2} - (\bar s)^2$ is the standard variance identity. A well-conditioned SSD has small $|s_{ij}|$ for all off-diagonal pairs.

\citet{weese2021strategies} introduced $\Var(s_+)$, the variance of the positive-cone inner products, as a design criterion for two-level SSDs. Their formulation minimizes $\Var(s_+) = \overline{s^2} - (\bar s)^2$ subject to a mean-positivity constraint $\bar s > 0$ and an efficiency constraint
$\overline{s^2}^* / \overline{s^2} \ge 0.8$,
where $\overline{s^2}^*$ is computed from the augmented model matrix $L = (\mathbf{1} \,|\, X)$. 
The same ideal, all column pairs sharing a positive correlation, is connected to sign recovery by \citet[Theorem 2]{stallrich2025optimal}. Under their known-sign construction, a symmetric two-level SSD with common positive column-pair correlations is optimal for lasso sign recovery when such a construction is available. We use this as a qualitative target and calibrate the squared mean
off-diagonal inner product using $(s^\star)^2$, where
$s^\star=\sqrt{n(p-n)/(p-1)}$ is the Welch lower bound on $\max_{i\ne j}|s_{ij}|$ for $p>n$.

We propose a simpler criterion calibrated by the \citet{welch1974lower} lower bound on the maximum off-diagonal inner product. For $p$ columns of
$\pm 1$ entries (each with squared norm $n$), Welch's bound states
\begin{equation}
  \max_{i \ne j} s_{ij}^2
  \;\ge\;
  \frac{n(p - n)}{p - 1}, \qquad p > n,
  \label{eq:welch_bound_intro}
\end{equation}
with equality when $|s_{ij}|=s^\star$ for every $i\ne j$.
We define the inverse squared scale:
\begin{equation}
  \eta_{\mathrm{W}}(n, p) =
  \begin{cases}
    \displaystyle \frac{p - 1}{n(p - n)}, & p > n, \\[6pt]
    1, & p \le n,
  \end{cases}
  \label{eq:eta}
\end{equation}
chosen so that, when every $s_{ij}$ equals the Welch saturating magnitude, $\eta_{\mathrm{W}}(n,p)\,(\bar s)^2 = 1$ exactly (Appendix \ref{app:welch_calibration}). For $p \le n$, the
saturated regime degenerates to the orthogonal-design limit, so the weight is set to $1$ and, on the positive cone, $\Var(s_+^{\mathrm{W}})$
reduces to the classical near-orthogonality objective
$E(s^2) = \overline{s^2}$ \citep{booth1962some}, minimized by an orthogonal design. The $\Var(s_+^{\mathrm{W}})$ criterion is
\begin{equation}
  \Var(s_+^{\mathrm{W}})
  \;=\;
  \underbrace{\overline{s^2} - (\bar s)^2}_{\text{variance term}}
  \;+\;
  \underbrace{\eta_{\mathrm{W}}(n,p)\,(\bar s)^2}_{\text{Welch-calibrated level penalty}}
  \;+\;
  \underbrace{\lambda \sum_{i < j} \min(s_{ij}, 0)^2}_{\text{positive-cone penalty}},
  \label{eq:welch_crit}
\end{equation}
with $\lambda = 10^6$ a large quadratic penalty driving the
off-diagonals toward the positive cone $\{s_{ij} \ge 0 : i \ne j\}$.
Empirically, $89$--$98\%$ of recovered designs across the simulation grid satisfy $\min_{i < j} s_{ij} \ge 0$ at coordinate-exchange convergence (Appendix \ref{app:cone_validation}). The criterion is minimized by coordinate exchange. The round-by-round implementation is given in Appendix \ref{app:criterion}.

\paragraph{Calibration parameter, not a tuning parameter.} The Welch weight $\eta_{\mathrm{W}}$ is determined by $(n,p)$.
We use $\lambda=10^6$ as the default positive-cone penalty weight. Across the five weights $\{10^4,\ldots,10^8\}$, the reported metrics vary little at the setting examined in Appendix \ref{app:lambda_sweep}.

For $p\le n$, $\eta_{\mathrm{W}}=1$ and the first two terms of
the criterion reduce to $E(s^2)$. The positive-cone penalty discourages negative inner products, while the level penalty penalizes $(\bar s)^2$ on the Welch scale. A design cannot in general meet the Welch bound with all-positive off-diagonals when $p>n$, so the criterion acts as a Welch-scaled positive-cone
heuristic.

\paragraph{Positive cone versus near-orthogonality.}
Classical SSD construction targets near-orthogonality by minimizing
$E(s^2) = \overline{s^2}$ \citep{booth1962some} or $|s|_{\max}$ alone, accepting mixed-sign $s_{ij}$ as a natural consequence. The positive-cone preference of $\Var(s_+^{\mathrm{W}})$ is more aligned
with the LassoCV importance scoring used in Stage 1. If an inactive factor is correlated with an active one through a mixed-sign correlation, the LassoCV fit can give the inactive factor a spurious nonzero coefficient. Keeping the correlations all the same sign makes the fit more likely to shrink that spurious coefficient back to zero. This is a heuristic alignment between design and analysis. In the criterion ablation in Appendix \ref{app:abl}, the Welch criterion attains the lowest end-to-end LOSS, ahead of the best retained alternative by $0.65$ and $0.32$ in the two simulated cases.

\paragraph{Criterion scope.}
The criterion as stated controls inner products between the main-effect columns. Quadratic effects are unidentifiable on the $\pm 1$ Stage 1 design, where every squared column $x_j^2 \equiv 1$ coincides with the intercept, so they are deferred to Stage 2. A known active interaction can instead be protected by augmenting the model matrix with its $\pm 1$ product column (the $\varphi^+$ extension, with formula and design-quality pilot in Appendix \ref{app:phiplus}).

\subsection{The Two-Stage Procedure}
\label{sec:procedure}

With the model and the design criterion in place, we now describe the two stages.

\subsubsection{Stage 1: Sequential Screening with Graduated Quantile
  Elimination}
\label{sec:stage1}

\begin{figure}[!htbp]
\centering
\begin{tikzpicture}[
    >=Latex,
    box/.style={rectangle, rounded corners=2pt, draw=blue!55!black,
      thick, fill=blue!6, minimum width=22mm, minimum height=12mm,
      align=center, font=\footnotesize, inner sep=2pt},
    hand/.style={rectangle, rounded corners=2pt, draw=green!45!black,
      dashed, fill=green!4, align=center, font=\footnotesize, inner sep=3pt},
    arrow/.style={->, line width=0.7pt, draw=blue!55!black},
  ]
  \node[box] (a) {$p$ candidate\\factors};
  \node[box, right=8mm of a] (b) {Coordinate\\exchange design};
  \node[box, right=8mm of b] (c) {Fit model,\\score $|\hat\beta_j|$};
  \node[box, right=8mm of c] (d) {Graduated\\quantile cut};
  \node[box, right=8mm of d, fill=blue!12] (e) {Selected set\\$\hat\cA$ ($k_{\mathrm{sel}}$)};

  \draw[arrow] (a) -- (b);
  \draw[arrow] (b) -- (c);
  \draw[arrow] (c) -- (d);
  \draw[arrow] (d) -- (e);

  \draw[arrow, dashed] (d.north) -- ++(0,6mm) -| (b.north)
    node[pos=0.25, above, font=\scriptsize, text=black]
      {repeat over rounds $t = 1, 2, \ldots$ until the set is stable};

  \node[font=\footnotesize\bfseries, text=blue!55!black,
        below=4mm of b.south] {Stage 1 (sequential)};

  \node[hand, below=12mm of e, minimum width=22mm, minimum height=12mm] (f) {to Stage 2};
  \draw[arrow, draw=green!45!black, dashed] (e) -- (f);
\end{tikzpicture}
\caption{The Stage 1 screening loop of $\Sthree$.}
\label{fig:flow}
\end{figure}
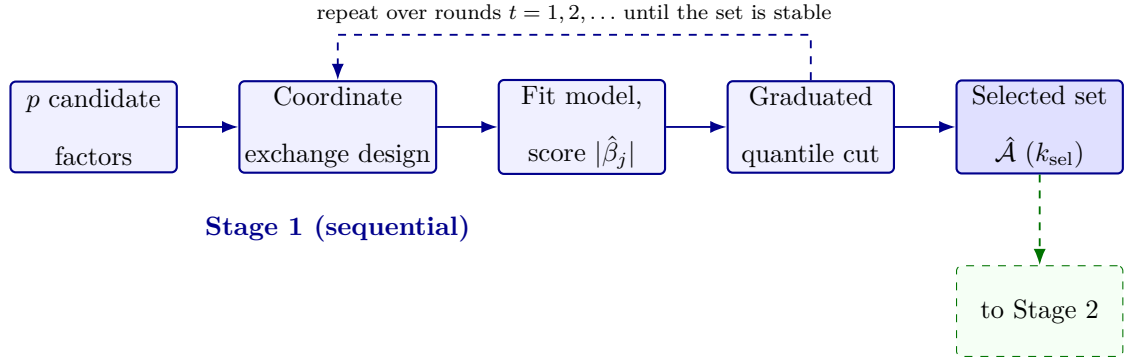

Stage 1 repeats a loop of screening and elimination over rounds $t = 1, 2, \ldots$. Figure~\ref{fig:flow} sketches the loop. Each screening round runs the following steps:
\begin{enumerate}[leftmargin=*, topsep=2pt, itemsep=2pt]
  \item \textbf{Allocate.} Round $t$ starts from the current candidate set $\cC_t \subseteq [p]$, initialized as $\cC_1 = [p]$, and its fixed subset $\cF_t \subseteq \cC_t$, initialized as $\cF_1 = \emptyset$, which holds the factors whose importance and sign of effect have already been committed. The open set $\cO_t = \cC_t \setminus \cF_t$ holds the candidates still under assessment. The round adds $n_t = \max(2, \lceil c\,|\cO_t| \rceil)$ new runs for a cost ratio $c \in (0, 1)$, bringing the cumulative run count to $N_t = \sum_{r \le t} n_r$.
  \item \textbf{Design.} The new runs are built by coordinate exchange \citep{meyer1995coordinate} under the $\Var(s_+^{\mathrm{W}})$ criterion of Section \ref{sec:welch}, evaluated on the cumulative design, all runs so far together with the new round, restricted to the current candidates, with $\eta_{\mathrm{W}}$ taken at the cumulative dimensions $(N_t, |\cC_t|)$. Only the new round's entries for open factors are exchanged. Earlier runs stay as collected, and each fixed factor enters every new run at its committed level, the $\pm 1$ value assigned when the factor was fixed (step 5). In practice, eliminated factors can be set to their default or cheapest levels. Here, we set them to zero in all subsequent runs.
  \item \textbf{Fit and score.} All $N_t$ runs are pooled and a regularized linear (LassoCV) model with round-block indicators $Z_{i,r} = \mathbf{1}[\text{run } i \text{ collected in round } r]$ is fitted:
\begin{equation}
  y_i \;=\; \alpha
        + \mathbf{x}_{\cC_t,i}^\top \bm{\beta}_{\cC_t}
        + \sum_{r=2}^{t} \gamma_r Z_{i,r}
        + e_i.
  \label{eq:surrogate}
\end{equation}
Each factor $j$ is scored by the absolute LassoCV coefficient
\begin{equation}
  I_j \;=\; |\hat\beta_j|,
  \label{eq:impscore}
\end{equation}
which vanishes whenever $\hat\beta_j = 0$, and the sign of each nonzero coefficient is recorded as that factor's current direction.
  \item \textbf{Threshold.} The elimination quantile at round $t$ follows the graduated quantile (GQ) rule $q_t = q \cdot \min(1, N_t / |\cC_t|)$, where $q \in (0, 1)$ is the base elimination quantile. 
  When $N_t < |\cC_t|$, $q_t < q$ and scales linearly with $N_t / |\cC_t|$, the runs collected per candidate, which helps protect against premature elimination, and the full rate $q$ is reached once $N_t / |\cC_t| \ge 1$. The threshold $\tau_t$ is the $q_t$-quantile of the scores of the currently open factors, and every open factor with $I_j \le \tau_t$ is marked for elimination. Whether a marked factor is actually removed is settled in step 5, where fixing takes precedence.
  \item \textbf{Fix and update.} The scores of all open factors are sorted decreasingly and the elbow $r^\star = \arg\max_r (I_{(r)} - I_{(r+1)})$ marks the largest gap between consecutive sorted scores. The top $r^\star$ open factors that have a recorded direction are committed to the fixed set, each locked at the $\pm 1$ level given by its direction. If only one open factor remains, it is committed directly. Marked factors outside the fixed set are then eliminated, giving the updated sets $\cC_{t+1}$ and $\cF_{t+1}$.
\end{enumerate}
Stage 1 stops when the fixed set does not grow during a round with $t \ge 2$, when the open set is empty, or when the round index reaches $S_{\max}$ (default $10$), and returns the final fixed set with its committed signs as the selected set $\hat\cA$, of size $k_{\mathrm{sel}} = |\hat\cA|$.

The cost ratio $c$ and the elimination quantile $q$ are the only operating choices in the pipeline, and Section \ref{sec:sensitivity} maps a safe operating region for them. The design criterion itself has no constant for the practitioner to set.

Figure \ref{fig:elimination_trace} illustrates the Stage 1 loop on one replication of a main-effect screening problem at $p = 80$,
$k_{\mathrm{true}} = 5$ active factors, and $(c, q) = (0.50, 0.65)$. The first round (Round 1) draws
$n_1 = \lceil 0.50 \cdot 80 \rceil = 40$ runs over the full pool $|\cC_1| = 80$ and applies the graduated rate $q_1 = 0.65 \cdot (40/80) = 0.325$. It fixes the two strongest actives, $x_0$ and $x_1$, and eliminates $64$ inactive factors. The realized count ($64$ of $80$) far exceeds the nominal $32.5\%$ because the LassoCV model drives most inactive coefficients to exactly zero in this round ($n_1 = 40 < |\cC_1| = 80$). The $0.325$-quantile of the importance scores is then itself zero, so every zero-scored factor is
cut at once. Round 2 adds $7$ runs over the surviving $16$ candidates at the full rate $q_2 = 0.65$, eliminating $9$ more and fixing $x_2$ and $x_3$. Round 3 adds $2$ runs over the remaining $7$
candidates, fixes the last active $x_4$ at the elbow, and removes the two surviving inactive factors. The run recovers all five actives
($\hat\cA = \cA$) with $N_{\mathrm{S1}} = 49$ Stage 1 runs. A one-shot LassoCV at the same $p = 80$ typically selects around $19$ factors, and a full second-order model in $19$ factors has $210$ coefficients (Section \ref{sec:stage2}).

\begin{figure}[!htbp]
\centering
\includegraphics[width=0.85\textwidth]{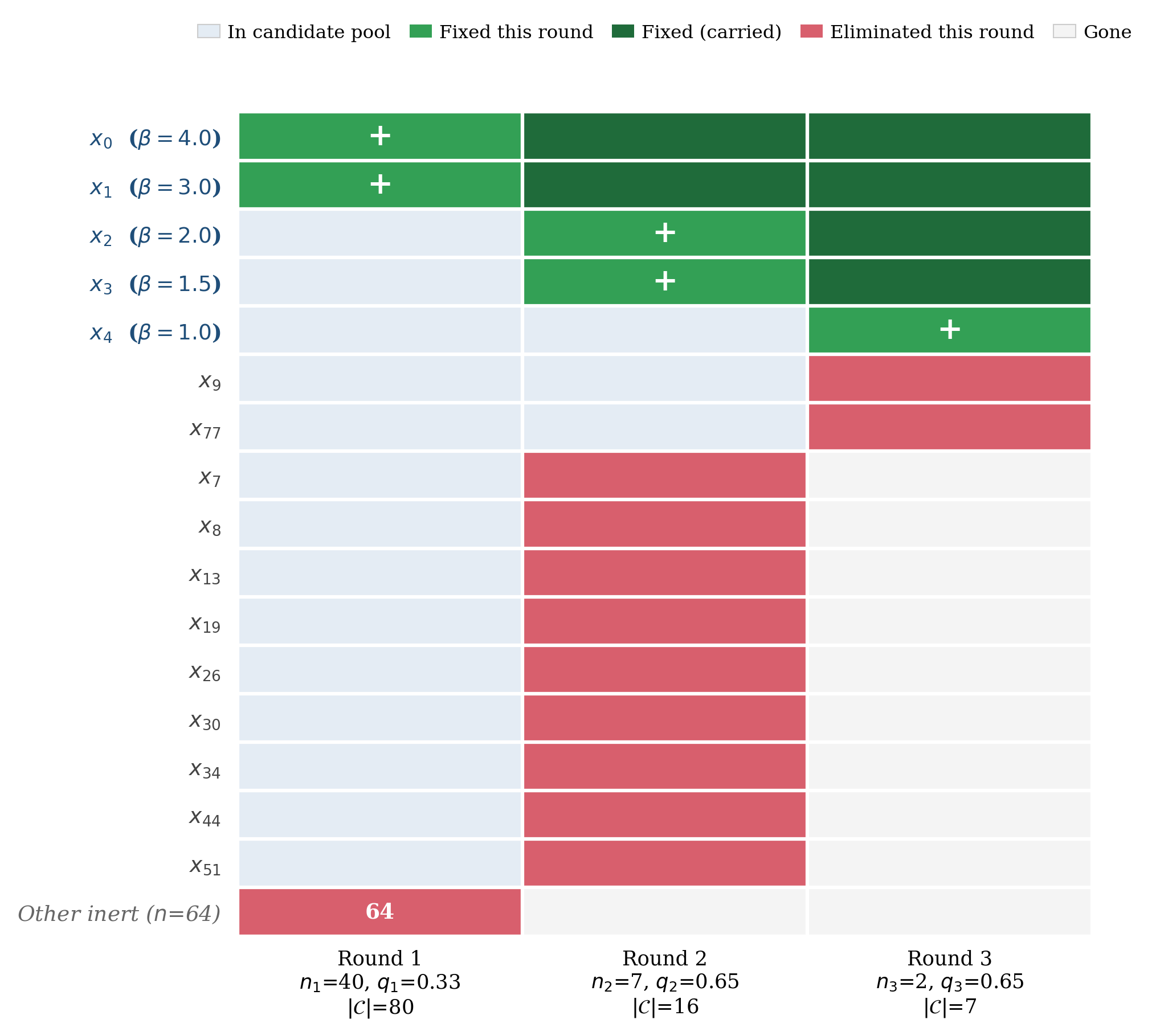}
\caption{Round-by-round status grid for one replication of a main-effect screening problem at $p = 80$, $k_{\mathrm{true}} = 5$,
$(c, q) = (0.50, 0.65)$, with true active factors $x_0$--$x_4$.}
\label{fig:elimination_trace}
\end{figure}

\begin{remark}[Scorer-agnostic design]
\label{rem:scorer}
The GQ elimination framework is the core algorithmic contribution and the importance scorer is a plug-in. The default $I_j = |\hat\beta_j|$ uses the absolute LassoCV coefficient and is adopted for its simplicity. The framework also accommodates a range of alternative scorers, from model-free dependence measures such as distance correlation \citep[dCor,][]{szekely2007measuring,li2012feature}
and Ball correlation \citep[BCor,][]{pan2020ball} to random forest importances such as SHAP \citep{lundberg2017unified} and Gini
\citep{breiman2001random}. Appendix \ref{app:scorer} demonstrates the swap on an interaction-dominated response, where the default linear model
fails and screening quality tracks the chosen scorer while the rest of the pipeline is held fixed.
\end{remark}

\paragraph{Mis-fixation diagnostic.}
Locking a factor at the wrong sign, or committing a noise factor, could in principle propagate errors. Early noise commitment is rare, but terminal noise inclusion is common in the weak-signal regime and is
bounded by the small final shortlist. Appendix \ref{app:misfix} shows that most noise commitment occurs only after the open inactive set has
already contracted.

\subsubsection{Stage 2: Response Surface Optimization}
\label{sec:stage2}

A full second-order model in the $k_{\mathrm{sel}}$ selected
factors has
\begin{equation}
  p_{\mathrm{poly}}(k_{\mathrm{sel}})
  = 1 + 2k_{\mathrm{sel}} + \binom{k_{\mathrm{sel}}}{2}
  \label{eq:ppoly}
\end{equation}
parameters. We reuse the Stage 1 runs on the selected factors when fitting the Stage 2 model. Larger selected sets typically require more Stage 2 runs. 

Standard one-shot defaults select much larger sets. Appendix \ref{app:baselines} shows that LassoCV and the other default one-shot rules return large shortlists across the sampled $n/p$ range, and that a stricter penalty does not repair this, since it trades Type I error for Type II error and can drop active factors that Stage 2 needs. A one-shot rule can be forced to return a fixed top-$k$ shortlist, but that cap must be chosen externally, before any data are seen. The $\Sthree$ shortlist size is instead set by the data, through the graduated elimination and the elbow rule. Beyond keeping the Stage 2 model small, $\Sthree$ selects far fewer false positives, shown in Section \ref{sec:sim}.

Let $\mathcal A_2\subseteq\hat\cA$ denote the factors carried into
Stage 2 after any capping or matching. Stage 2 is skipped when $|\mathcal A_2|<2$. Otherwise, let $r_1$ denote the rank of the
full second-order model matrix of the cumulative Stage 1 design restricted to $\mathcal A_2$. We augment this design with
$N_{\mathrm{S2}}=\max(0,p_{\mathrm{poly}}(|\mathcal A_2|)-r_1+5)$ new runs. Since $r_1\le p_{\mathrm{poly}}(|\mathcal A_2|)$, at least five runs are added.

We select these runs by greedy forward $D$-optimal augmentation. The candidate pool is the full $\{-1,0,+1\}^{|\mathcal A_2|}$ grid when $3^{|\mathcal A_2|}\le2000$, and otherwise consists of
2000 points drawn at random from this grid. At each step, we sample up to 300 points from the pool and add the one giving the largest log determinant of the augmented information matrix, with $10^{-8}$ added to each diagonal entry.

The combined Stage 1 and Stage 2 data are fitted by OLS with BIC-guided forward selection. For $\Sthree$, the Stage 1 block indicators (block history $\mathbf B$) and a Stage 2 indicator are retained unconditionally as nuisance terms. We retain the linear main effects of all factors in $\mathcal A_2$ to preserve effect hierarchy. BIC selection is applied only to interaction
and quadratic terms.

We omit the block terms from the fitted model to obtain the response surface $\hat\mu$ on $[-1,1]^{|\mathcal A_2|}$.
We use the stationary point when $\hat\mu$ has a negative definite Hessian and this point lies inside the box. Otherwise we use
box-constrained numerical optimization \citep[L-BFGS-B,][]{byrd1995limited}.
The paired comparisons in Section \ref{sec:metrics} use the same augmentation and polynomial selection rules at the matched factor count.

We form the full setting $\hat{\mathbf{x}}\in[-1,1]^p$ using the optimized values for the Stage 2 factors. Fixed factors omitted by capping or matching retain their committed levels, and all other omitted factors are set to zero.
We evaluate this setting using the optimality loss \eqref{eq:loss}. Appendix \ref{app:fullalg} gives the complete two-stage procedure.

\section{Theoretical Properties}
\label{sec:theory}

Two results support the two-stage design. The first bounds the probability that an inactive factor is ever committed into the selected set, separating the event that an inactive factor survives an elimination round from the event that the elbow rule commits it. The second bounds the total number of Stage 1 runs. The supporting elimination-rate proposition, an observed-size false discovery proportion (FDP) corollary, the round-level diagnostic table, and all proofs are in Appendix \ref{app:proofs}.

Let $\cA \subset [p]$ with $|\cA| = k$ denote the active set, $\cN = [p] \setminus \cA$ the inactive set, whose elements we also call noise factors, and $\cF \subseteq [p]$ the final fixed set returned by $\Sthree$ with $|\cF| = k_{\mathrm{sel}}$. The
procedure executes at most $S_{\max}$ rounds. At round $t$ the candidate set is $\cC_t$, the still-open inactive set is $\cO_t^{\cN} = \cN \cap \cO_t$, and $q_t$ and $\tau_t$ are the round's elimination rate and threshold of Section \ref{sec:stage1}. We write $\cG_t = \sigma(X_{1:t}, \mathbf{y}_{1:t})$ for the filtration generated by the first $t$ rounds of data. For an inactive factor $j$, let $H_t(j)$ denote the event that round $t$ is executed with $j$ open at its start and the score of $j$ exceeds $\tau_t$, and let $C_t(j)$ denote the event that round $t$ is executed and the elbow rule commits $j$. If Stage 1 stops before round $t$, neither event occurs. We can therefore sum over $t=1,\ldots,S_{\max}$ in the
probability bounds. We also write $\{j \in \cO_t\}$ for the event that round $t$ is executed with $j$ still open at its start.

\begin{assumption}[Open-set survival and commitment]
\label{as:noise}
For each inactive factor $j \in \cN$, each round $t$, and every event $B \in \cG_{t-1}$ satisfying $B \subseteq \{j \in \cO_t\}$ and
$\bbP(B)>0$,
\begin{equation}
  \bbP\bigl(H_t(j) \bigm| B\bigr)
  \;\le\; \rho_t,
  \label{eq:noise_surv}
\end{equation}
and
\begin{equation}
  \bbP\bigl(C_t(j) \bigm| B\bigr)
  \;\le\; \delta_t.
  \label{eq:commit_ctrl}
\end{equation}
The constants $\rho_t$ and $\delta_t$ may depend on the fixed design protocol, the budget rule, and the data-generating regime, but not on the identity of $j$ within the inactive set.
\end{assumption}

The survival part of Assumption \ref{as:noise} bounds the
probability that an individual inactive factor exceeds the
GQ threshold, conditional on the earlier data. The deterministic contraction result below bounds the total
number of open factors. The commitment part is separate. It controls the chance that an inactive factor that remains open is placed in the elbow prefix and locked into the downstream model. If every factor surviving the first round were committed, the first-round survival rate would be the best attainable
bound, so survival control alone cannot deliver multi-round decay. For the final fixed set $\cF$, the false discovery proportion is $\mathrm{FDP} = |\cF \cap \cN| / (|\cF| \vee 1)$, the fraction of selected factors that are noise.

\begin{proposition}[Round-by-round inactive-commitment bound]
\label{prop:inactive_commitment}
Under Assumption \ref{as:noise}, with rounds capped at $S_{\max}$ and unexecuted rounds contributing empty events, for any inactive factor $j \in \cN$,
\begin{equation}
  \bbP(j \in \cF)
  \;\le\;
  \sum_{t=1}^{S_{\max}} \delta_t \prod_{s<t} \rho_s.
  \label{eq:inactive_commitment}
\end{equation}
Consequently,
\begin{equation}
  \bbE\bigl[|\cF \cap \cN|\bigr]
  \;\le\;
  (p - k) \sum_{t=1}^{S_{\max}} \delta_t \prod_{s<t} \rho_s.
  \label{eq:efp}
\end{equation}
The bound controls inactive-factor commitment into Stage 2 only.
\end{proposition}

\begin{remark}[Clean ladder as a special case]
If $\delta_t = 0$ for every $t < s_0$, the sum in \eqref{eq:inactive_commitment} starts at $s_0$. An inactive factor can then be committed only after surviving all rounds before $s_0$, so each remaining term includes the survival product $\prod_{s<s_0}\rho_s$. The diagnostics in Appendix \ref{app:noise_diagnostics} report
small but nonzero first-round commitment rates. The proposition therefore retains the general $\delta_t$ terms.
\end{remark}

\begin{proposition}[Stage 1 run bound]
\label{prop:budget}
Under GQ screening with the default LassoCV scorer and parameters $c, q \in (0, 1)$ starting from
$|\cO_1| = p \ge 2$ open candidates, the total number of Stage 1 runs over $S$ executed rounds satisfies $N_{\mathrm{S1}} \le p / q + 2S$, with $S \le S_{\max}$. For fixed $c$ and $q$, the round bound in Appendix \ref{app:proofs} gives $S=O(\log p)$, and hence $N_{\mathrm{S1}}\le p/q+O(\log p)$.
\end{proposition}

\section{Simulation Study}
\label{sec:sim}

\subsection{Simulation Setup}
\label{sec:setup}

We study two signal configurations, both with $|\cA| = 5$ active factors ($x_0, \ldots, x_4$) among $p$ candidates and the remaining factors inactive. \textbf{Case 1} has strong main effects
($\beta_{\max} = 4.0$, main-effect dominated). \textbf{Case 2} has weaker main effects ($\beta_{\max} = 2.0$, weak-signal regime). Both share the same interaction and quadratic structure
($\beta_{01} = 2$, $\beta_{12} = -1.5$, $\beta_{00} = -1.5$, $\beta_{33} = 1$). 
At $p=80$, we also vary $k_{\mathrm{true}}$ from $4$ to $12$ in both cases. Each setting uses the first $k_{\mathrm{true}}$ main effect coefficients and the corresponding interaction and quadratic terms specified in Appendix \ref{app:simulation_details}. This examines how screening and optimization performance change with the number of active factors.

We run $p \in \{40, 80, 120, 200\}$ candidates at
$k_{\mathrm{true}} = 5$, plus the sparsity sweep above at $p = 80$, with Gaussian noise $\sigma = 1.0$ and a batch-effect standard deviation $\sigma_b = 1.0$ at batch size 10. The two pipeline operating parameters span a $7 \times 9 = 63$-cell grid, $c \in \{0.35, \ldots, 0.65\}$ (7 values) and $q \in \{0.40, \ldots, 0.80\}$ (9 values), per $(k, p)$. Case 1 runs four arms: one is $\Sthree$, and the other three are one-shot LassoCV arms whose two-level designs are built with the $\Var(s_+^{\mathrm{W}})$, $\Var(s_+)$, and $E(s^2)$ design criteria, tagged l\_welch, l\_varsplus, and l\_es2, as distinct from the round-by-round LassoCV model inside $\Sthree$. Case 2 drops l\_es2 and runs three arms. The $\Sthree$ arm uses $\Var(s_+^{\mathrm{W}})$ throughout. These criteria are the one-shot design comparators retained after the criterion ablation in Appendix \ref{app:criterion}. We use $S_{\max} = 10$ rounds and $R = 500$ replications per cell.

\subsection{Evaluation Metrics and Protocol}
\label{sec:metrics}

We evaluate each arm on two axes. Screening quality uses
Type I error, Type II error, the $F_1$ score, and the Matthews
correlation coefficient (MCC). Optimization quality uses the
LOSS of \eqref{eq:loss}, written $\mathrm{LOSS}_{\Sthree}$ for $\Sthree$ and $\mathrm{LOSS}_{\mathrm{l\_welch}}$ for the l\_welch baseline (and analogously for the other LassoCV arms). The main comparison uses the paired gap
$\Delta\mathrm{LOSS}=\mathrm{LOSS}_{\Sthree}
-\mathrm{LOSS}_{\mathrm{l\_welch}}$
at the matched factor count defined below.
For the heatmaps, $\Sthree$ uses its own selected set capped at
15 factors, and each LassoCV arm uses its corresponding matched
factor count. Within each parameter cell, we calculate
$\Delta\mathrm{LOSS}_{\min}
=\overline{\mathrm{LOSS}}_{\Sthree}
-\min_{\mathrm{La}}\overline{\mathrm{LOSS}}_{\mathrm{La}}$,
where each mean uses that arm's completed Stage 2 optimizations.
Screening summaries use all $R=500$ replications.
Paired gaps require completion by both arms.
With more inactive than active factors, $F_1$ and MCC provide
complementary summaries of screening performance, with formulas
in Appendix \ref{app:simulation_details}.

In each replication $\Sthree$ sets its Stage 1 run count
$N_{\mathrm{S1}}$ adaptively. Each LassoCV arm receives a two-level design
of the same size, built by coordinate exchange with the same criterion.
Screening quality is computed on the raw selections.
For the paired comparisons we use the matched size
$k_{\mathrm{m}}=\min(k^{\Sthree}_{\mathrm{sel}},k^{\mathrm{La}}_{\mathrm{sel}},15)$,
retaining each arm's highest-scoring factors without refitting the screen.
The fixed cap of 15 limits the size of the fitted second-order surface.
Both arms use the same Stage 2 augmentation and polynomial
selection rules on their $k_{\mathrm{m}}$ factors.
Block indicators are included for $\Sthree$ and omitted for
the one-shot arms.
This matches the Stage 1 run count and the Stage 2 factor count,
while the additional runs can differ because the two Stage 1
model matrices can have different ranks $r_1$.
Appendix \ref{app:fixed_budget} adds a complementary comparison under
a fixed total budget, with each arm allocating its runs between screening
and the surface.

\subsection{Sensitivity Analysis}
\label{sec:sensitivity}

We examine performance over the $(c,q)$ grid before presenting
aggregate comparisons in the recommended operating ranges.
Figure \ref{fig:heatmaps} shows the mean of the
$\Delta\mathrm{LOSS}_{\min}$ gaps over $k=4,\ldots,12$
on the full $7\times9$ $(c,q)$ grid at $p=80$.

\begin{figure}[!htbp]
\centering
\includegraphics[width=\textwidth]{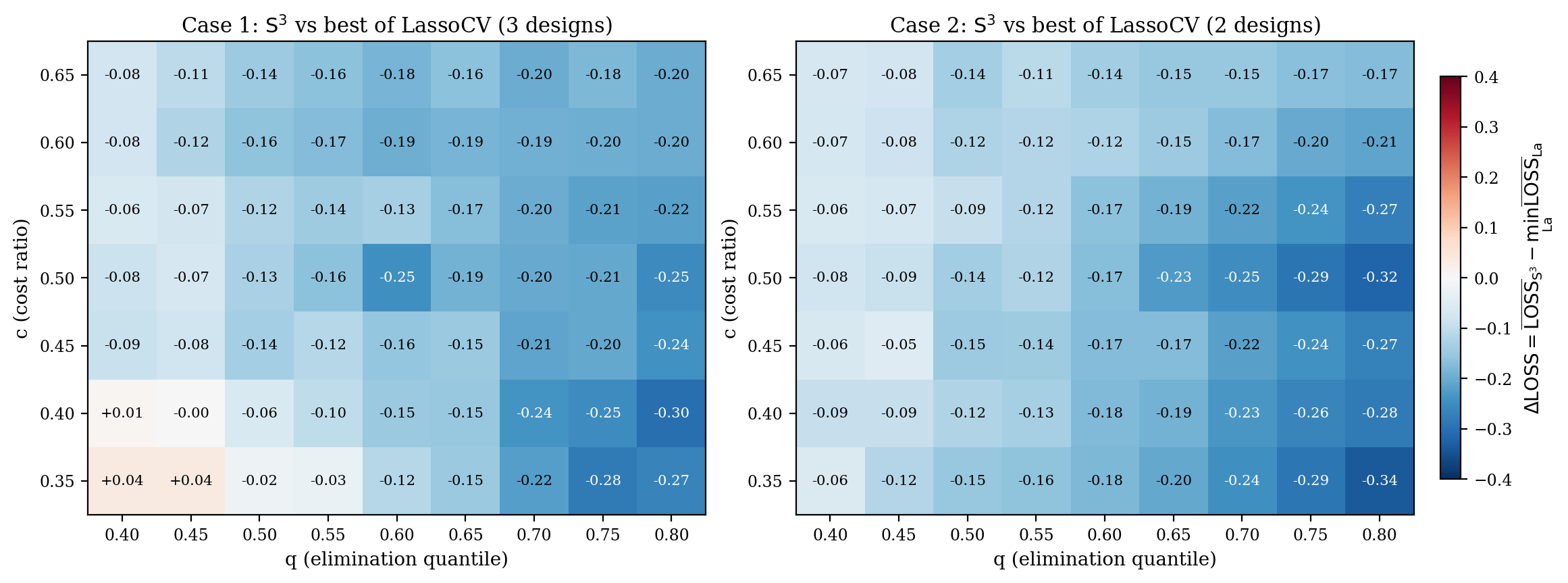}
\caption{Cell-mean gaps
$\Delta\mathrm{LOSS}_{\min}
=\overline{\mathrm{LOSS}}_{\Sthree}
-\min_{\mathrm{La}}\overline{\mathrm{LOSS}}_{\mathrm{La}}$
on the $7\times9$ $(c,q)$ grid at $p=80$.
The gaps are formed separately at each $k$ and then averaged
over $k=4,\ldots,12$.
Left panel: Case 1, with the minimum over l\_welch,
l\_varsplus and l\_es2.
Right panel: Case 2, with the minimum over l\_welch and l\_varsplus.
Negative values favor $\Sthree$ and positive values favor
one-shot LassoCV.}
\label{fig:heatmaps}
\end{figure}

$\Sthree$ has lower mean LOSS than the best LassoCV baseline
in 60 of 63 cells in Case 1, with a mean
$\Delta\mathrm{LOSS}_{\min}$ of $-0.150$, and in all 63 cells
in Case 2, with a mean gap of $-0.164$.
The recommended operating ranges vary with $p$, as discussed
in Section \ref{sec:main_comparison}.

The heatmap averages over a sweep of the true sparsity $k_{\mathrm{true}} \in \{4, \ldots, 12\}$ at $p = 80$ (sparsity level $k/p$ from $5\%$ to $15\%$). Table \ref{tab:ksweep} unpacks that sweep so the averaging is transparent. The qualitative picture is stable across the whole range. $\Sthree$'s selected size stays near $6$ while LassoCV's raw selection climbs from $18$ to $23$ in Case 1 (and from $15$ to $18$ in Case 2), so the second-order model implied by the raw one-shot selection grows further as the truth gets denser. The paired LOSS advantage holds at every $k_{\mathrm{true}}$ and widens for $k_{\mathrm{true}} \ge 6$, and
$\Sthree$'s $F_1$ leads throughout.

\begin{table}[!htbp]
\centering
\caption{Sensitivity of the $\Sthree$-versus-l\_welch comparison to
the true sparsity $k_{\mathrm{true}}$, at $p = 80$ over the recommended
band ($c \in [0.45, 0.65]$, $q \in [0.40, 0.80]$), $R = 500$
replications per cell. The $k_{\mathrm{sel}}$ and $F_1$
columns report $\Sthree\,/\,$LassoCV. $k_{\mathrm{sel}}$ is the raw selected size, $F_1$ uses the raw LassoCV selection, and
$\Delta\mathrm{LOSS}$ is the paired $\Sthree$-minus-l\_welch gap at matched $k_{\mathrm{m}}$.}
\label{tab:ksweep}
\small
\setlength{\tabcolsep}{5pt}
\begin{tabular}{c rcc rcc}
\toprule
 & \multicolumn{3}{c}{Case 1} & \multicolumn{3}{c}{Case 2} \\
\cmidrule(lr){2-4}\cmidrule(lr){5-7}
$k_{\mathrm{true}}$
  & $\Delta$LOSS & $k_{\mathrm{sel}}$ & $F_1$
  & $\Delta$LOSS & $k_{\mathrm{sel}}$ & $F_1$ \\
\midrule
4  & $-0.04$ & $5.8/17.6$ & $0.77/0.37$ & $-0.09$ & $5.9/14.9$ & $0.54/0.34$ \\
5  & $-0.10$ & $6.0/19.7$ & $0.79/0.39$ & $-0.08$ & $5.9/16.0$ & $0.53/0.36$ \\
6  & $-0.16$ & $6.1/21.0$ & $0.75/0.42$ & $-0.16$ & $6.1/16.4$ & $0.49/0.36$ \\
7  & $-0.25$ & $6.3/21.9$ & $0.71/0.43$ & $-0.17$ & $6.3/16.9$ & $0.46/0.36$ \\
8  & $-0.22$ & $6.5/22.8$ & $0.66/0.42$ & $-0.18$ & $6.5/17.6$ & $0.43/0.35$ \\
9  & $-0.22$ & $6.5/22.9$ & $0.62/0.43$ & $-0.24$ & $6.5/17.9$ & $0.41/0.35$ \\
10 & $-0.16$ & $6.6/23.0$ & $0.58/0.43$ & $-0.20$ & $6.5/17.3$ & $0.39/0.35$ \\
11 & $-0.18$ & $6.5/23.2$ & $0.55/0.43$ & $-0.24$ & $6.4/17.6$ & $0.37/0.35$ \\
12 & $-0.15$ & $6.6/23.2$ & $0.52/0.43$ & $-0.23$ & $6.5/17.8$ & $0.36/0.36$ \\
\bottomrule
\end{tabular}
\end{table}

\subsection{Main Comparison and Operating Point Recommendation}
\label{sec:main_comparison}

The $\Var(s_+^{\mathrm{W}})$ criterion makes the cost ratio $c$ the operating point parameter that matters. The elimination quantile $q$ is far less sensitive to $p$ and tolerates the full band $[0.40, 0.80]$ without retuning. The cells of largest $\Sthree$ LOSS advantage shift with $p$, and the mechanism is the first-round size $\lceil c\,p \rceil$, which scales with $p$. At small $p$ a larger $c$ is needed so each round carries enough runs for a stable importance ranking. At large $p$ even a small $c$ already yields a large absolute round, so a smaller $c$ suffices and avoids over-spending Stage 1. The quantile $q$ behaves similarly across its band: too small a $q$ eliminates too slowly to converge within $S_{\max}$ rounds, while too large a $q$ removes survivors faster than the LassoCV model can confirm them and risks cutting active factors, so a mid-band value avoids both.

As a rule of thumb, take $c$ relatively large when $p$ is small, and relatively small when $p$ is large, with $q$ left in the band
$[0.40, 0.80]$. This rule reads off $p$ alone, which is fixed and
known before any run is collected, and needs neither the unknown true sparsity $k_{\mathrm{true}}$ nor any LassoCV output as an input. The band, and the operating region it sits in, are calibrated from $\Sthree$'s own LOSS surface over $(c, q)$, not from the $\Sthree$-versus-LassoCV LOSS gap, so it is a prospective operating choice rather than post hoc tuning. As concrete starting points, an interior $c \in [0.55, 0.65]$ when $p$ is around 40, $c \in [0.45, 0.65]$ when $p$ is near 80, and $c \in [0.35, 0.50]$ in the high-dimensional regimes we simulate ($p \in \{120, 200\}$) all sit in the cells of largest $\Sthree$ LOSS advantage across both signal regimes. These are reliable starting points, not necessary conditions: a larger $c$ at $p = 80$, for instance, still beats LassoCV, just by a smaller LOSS margin. Paired comparisons within the recommended band are reported in Appendix \ref{app:matched}. Table \ref{tab:main_results} reports screening quality against l\_welch inside the recommended $c$-band per $p$, and Figure \ref{fig:loss_forest} reports the paired optimization gap $\Delta\mathrm{LOSS} = \overline{\mathrm{LOSS}}_{\Sthree} - \overline{\mathrm{LOSS}}_{\mathrm{l\_welch}}$ at matched $k_{\mathrm{m}}$.

\begin{table}[H]
\centering
\caption{Screening performance of $\Sthree$ and one-shot
LassoCV (La, l\_welch), averaged over the recommended $c$-bands
and $q\in[0.40,0.80]$ at $k_{\mathrm{true}}=5$.
Each parameter cell uses $R=500$ replications.
All metrics use raw selections. $F_1$ and MCC are calculated
from each cell's mean confusion counts before averaging
across cells.}
\label{tab:main_results}
\small
\begin{tabular}{c c c c c c}
\toprule
 & & \multicolumn{4}{c}{$\Sthree$/La} \\
\cmidrule(lr){3-6}
Case & $p$ & Type I error & Type II error & $F_1$ & MCC \\
\midrule
1 & 40  & $0.0428 / 0.3029$ & $0.207 / 0.047$
  & $0.76 / 0.47$ & $0.72 / 0.44$ \\
1 & 80  & $0.0219 / 0.1975$ & $0.137 / 0.026$
  & $0.79 / 0.39$ & $0.78 / 0.44$ \\
1 & 120 & $0.0173 / 0.1517$ & $0.134 / 0.027$
  & $0.77 / 0.36$ & $0.76 / 0.42$ \\
1 & 200 & $0.0095 / 0.1030$ & $0.085 / 0.010$
  & $0.81 / 0.33$ & $0.81 / 0.42$ \\
\midrule
2 & 40  & $0.0715 / 0.2386$ & $0.504 / 0.309$
  & $0.50 / 0.41$ & $0.43 / 0.33$ \\
2 & 80  & $0.0406 / 0.1639$ & $0.425 / 0.249$
  & $0.53 / 0.36$ & $0.50 / 0.36$ \\
2 & 120 & $0.0306 / 0.1296$ & $0.429 / 0.256$
  & $0.50 / 0.31$ & $0.48 / 0.34$ \\
2 & 200 & $0.0195 / 0.0910$ & $0.367 / 0.209$
  & $0.53 / 0.30$ & $0.53 / 0.35$ \\
\bottomrule
\end{tabular}
\end{table}

\begin{figure}[!htbp]
\centering
\includegraphics[width=0.75\textwidth]{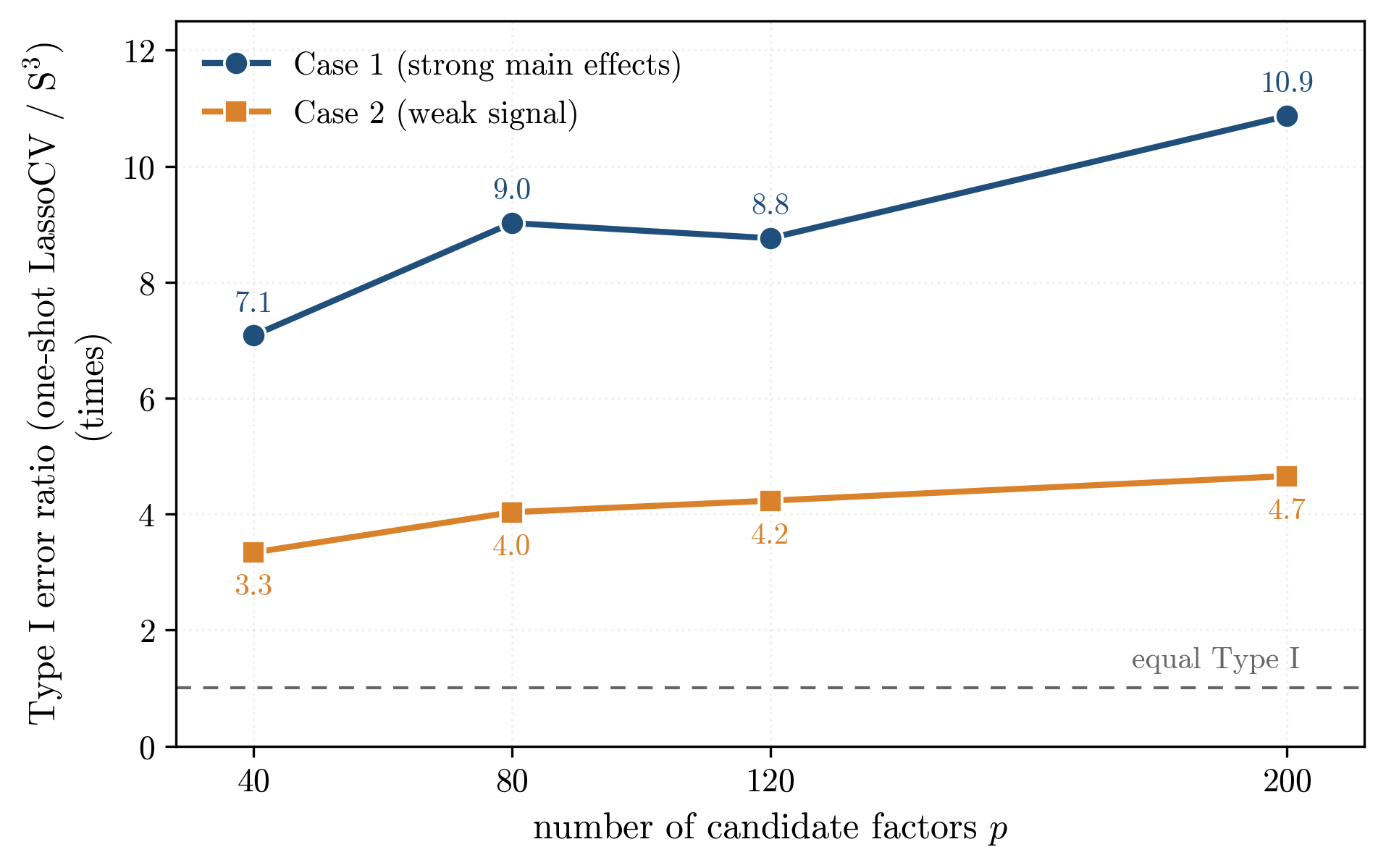}
\caption{Ratio of mean Type I error rates
(one-shot LassoCV$/\Sthree$) over the recommended $c$-band
and $q\in[0.40,0.80]$ at each $p$, with $k_{\mathrm{true}}=5$,
for both cases.}
\label{fig:typeI_ratio}
\end{figure}

The Type I ratio generally grows with $p$ (Figure \ref{fig:typeI_ratio})
because $\Sthree$'s Type I shrinks faster than that of LassoCV as the inactive pool grows, so the raw LassoCV shortlist produced by the default tuning becomes more costly at larger $p$.

Two points follow. First, $\Sthree$ has a strictly higher $F_1$ in all eight settings, $1.2$ to $2.4$ times that of one-shot LassoCV, and a higher MCC, $1.3$ to $1.9$ times. With an active fraction of $5/p$, the Type I gap of $3.3$ to $10.9$ times is what drives these gains. The lower Type I error comes with higher Type II error in all eight settings.

\begin{figure}[!htbp]
\centering
\includegraphics[width=0.85\textwidth]{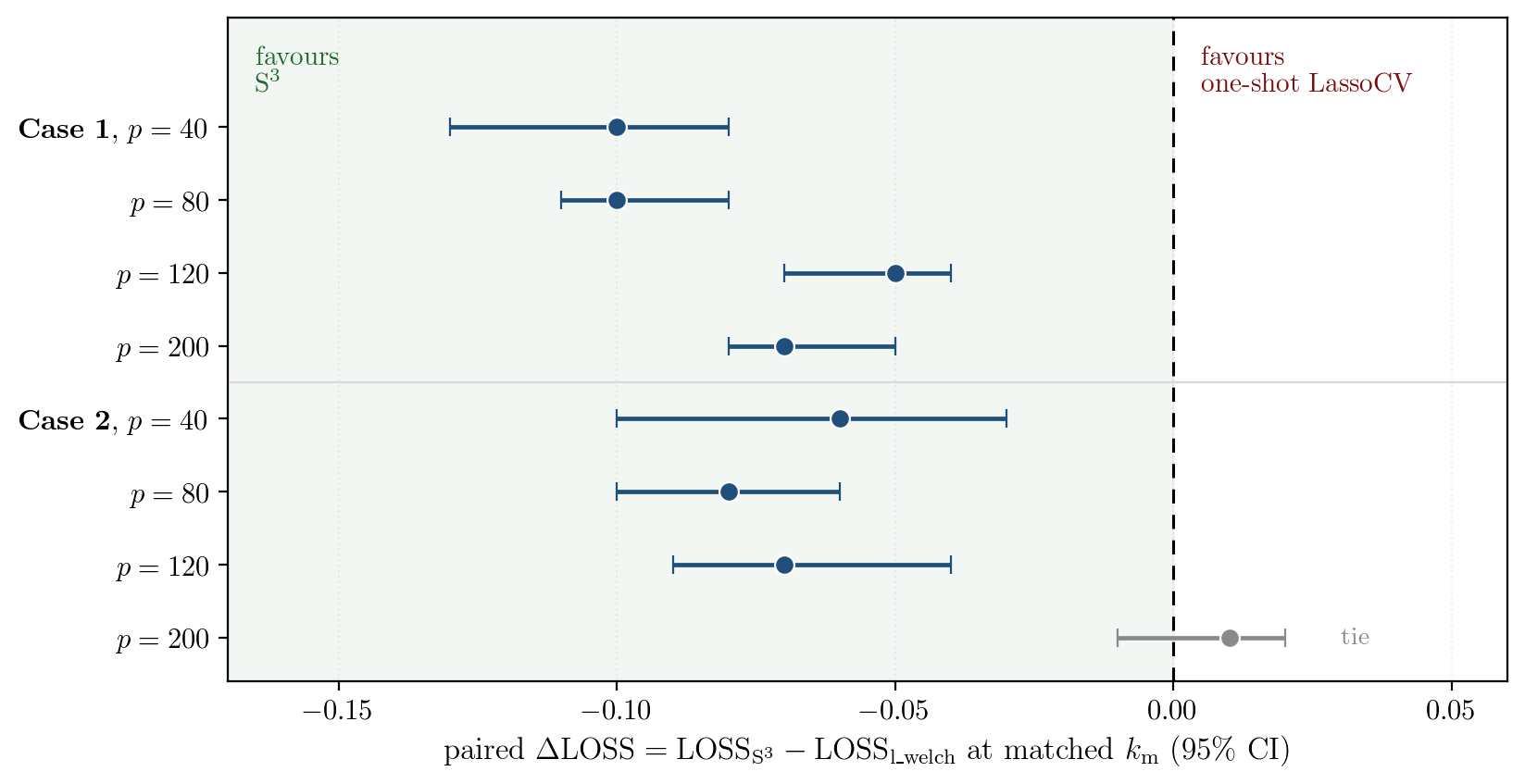}
\caption{Paired $\Delta\mathrm{LOSS}$ between $\Sthree$ and the one-shot
LassoCV baseline (l\_welch) at matched $k_{\mathrm{m}}$, with 95\%
confidence intervals per $(\text{case}, p)$ cell. An interval left of zero favors $\Sthree$.}
\label{fig:loss_forest}
\end{figure}

Second, optimization quality improves under the same comparison. $\Sthree$ holds a significant LOSS advantage in seven of the eight $(\text{case}, p)$ rows (from $-0.10$ to $-0.05$), with the eighth (Case 2, $p = 200$) a statistical tie.

\subsection{Computational Cost}
\label{sec:timing}

Table \ref{tab:timing} reports Stage 1 wall times by $p$ for Case 1 at one $(c, q)$ choice per $p$ inside the recommended band, with $R = 500$ replications, each on one core of an AMD Zen 4 node. The ratio of median wall times $t(\Sthree)/t(\mathrm{LassoCV})$ ranges from $0.85$ to $0.95$, so $\Sthree$ is faster at every tested size.

\begin{table}[!htbp]
\centering
\caption{Stage 1 wall time of $\Sthree$ versus the one-shot LassoCV
baseline at the matched Stage 1 budget $n = N_{\mathrm{S1}}$, across
$p \in \{40, 80, 120, 200\}$ at $k_{\mathrm{true}} = 5$,
Case 1, $R = 500$. Both timings cover design construction and
model fitting. Times are medians over replications, the ratio is
the ratio of medians ($t_{\Sthree}/t_{\mathrm{La}}$), and $N_{\mathrm{S1}}$, $n_{\mathrm{rounds}}$
(screening rounds), and $k_{\mathrm{sel}}$ ($\Sthree$/La selected
size) are means.}
\label{tab:timing}
\small
\setlength{\tabcolsep}{4.0pt}
\begin{tabular}{c c c r r c r r r}
\toprule
$p$ & $c$ & $q$
  & $t_{\Sthree}$ med.\
  & $t_{\mathrm{La}}$ med.\
  & ratio 
  & $N_{\mathrm{S1}}$
  & $n_{\mathrm{rounds}}$
  & $k^{\Sthree}_{\mathrm{sel}}$ / $k^{\mathrm{La}}_{\mathrm{sel}}$ \\
\midrule
 40 & 0.60 & 0.65 &   0.59\,s &   0.63\,s & $0.93$ &  34.8 & 4.00 & 5.1 / 15.1 \\
 80 & 0.55 & 0.65 &   6.83\,s &   7.17\,s & $0.95$ &  58.1 & 4.22 & 5.7 / 19.1 \\
120 & 0.45 & 0.65 &  32.36\,s &  35.20\,s & $0.92$ &  66.9 & 4.29 & 5.9 / 22.0 \\
200 & 0.45 & 0.65 & 195.93\,s & 229.17\,s & $0.85$ & 106.2 & 4.55 & 6.0 / 25.2 \\
\bottomrule
\end{tabular}
\end{table}

\FloatBarrier
\section{Benchmark Application: Borehole Function}
\label{sec:borehole}

The Borehole function \citep{harper1983sensitivity} models the steady-state groundwater flow rate $Q$ through a borehole as a function of 8 physical parameters: borehole radius ($r_w$), radius of influence ($r$), upper aquifer transmissivity ($T_u$), upper aquifer head ($H_u$), lower aquifer transmissivity ($T_l$), lower aquifer head ($H_l$), borehole length ($L$), and hydraulic conductivity ($K_w$). The upper aquifer head $H_u$ has a positive effect on the
flow rate $Q$, whereas the lower aquifer head $H_l$ has a
negative effect. Following \citet{moon2012two}, we embed the 8 active factors among $p = 80$ candidates (72 inactive factors), creating a supersaturated screening problem with known ground truth. Each physical parameter is mapped to $\{-1, +1\}$ at its documented low/high physical range. Responses are generated by evaluating the Borehole function at the requested settings of the eight physical factors and adding the same Gaussian noise and batch effects as in the simulation study, following the observation model \eqref{eq:obs}. The 72 embedded inactive factors do not enter the function. The oracle optimum is $f^\star = 309.58$, computed by differential evolution over $[-1, 1]^8$.

A benchmark function is the right testbed for a sequential method. The design at each round depends on the previous round's analysis, so the runs must be generated live, which a static published design-and-response table cannot provide. The Borehole function can be evaluated at the factor settings
requested by the procedure.

We apply $\Sthree$ with the same pipeline as the simulation study
(4-arm pairwise: $\Sthree$ vs l\_welch, l\_varsplus, l\_es2) over the full $7 \times 9 = 63$ $(c, q)$ grid at $R = 500$ replications per cell.

\begin{table}[t]
\centering
\caption{Borehole benchmark ($p = 80$, $k_{\mathrm{true}} = 8$,
$R = 500$): $\Sthree$ versus one-shot LassoCV paired-replication comparison over all $63$ $(c, q)$ cells. Paired $\Delta\mathrm{LOSS}$ 95\% CIs computed from replicate-level LOSS differences across all replications that reach Stage 2.
$k_{\mathrm{raw}}$ is the mean raw selected size ($\Sthree$/La). The
$F_1$ and MCC ratios are $\Sthree$/La, and the Type I (T1) ratio is
La/$\Sthree$.}
\label{tab:borehole}
\small
\setlength{\tabcolsep}{5pt}
\begin{tabular}{l c c c c c c}
\toprule
$\Sthree$ vs & paired $\Delta$LOSS (95\% CI) & cells won
  & $k_{\mathrm{raw}}$ ($\Sthree$/La) & $F_1$ ratio & MCC ratio & T1 ratio \\
\midrule
l\_welch    & $-5.08~[-5.40, -4.75]$ & $59/63$ & $5.4 / 19.7$ & $1.62$ & $1.90$ & $11.2$ \\
l\_varsplus & $-2.62~[-2.92, -2.31]$ & $57/63$ & $5.4 / 17.8$ & $1.53$ & $1.76$ & $9.8$  \\
l\_es2      & $-2.84~[-3.15, -2.53]$ & $57/63$ & $5.4 / 17.7$ & $1.53$ & $1.75$ & $9.7$  \\
\bottomrule
\end{tabular}
\end{table}

From Table \ref{tab:borehole}, we found that $\Sthree$ wins on $59/63$ cells against the l\_welch baseline, with a paired-replication mean LOSS advantage of $-5.08$ (95\% CI $[-5.40, -4.75]$). $F_1$ and MCC dominance is complete on $63/63$ cells against every LassoCV baseline. Aggregated over the $(c, q)$ grid against the best-performing of the three LassoCV baselines at each cell, $\Sthree$ wins $57/63$ cells with mean $\Delta\mathrm{LOSS} = -2.25$ (Appendix \ref{app:borehole}). The few exceptions sit in the low-elimination corner ($c \le 0.45$, $q \le 0.45$), outside the recommended operating region and consistent with the less stable behavior seen there in the simulation study.

\paragraph{Factor-level interpretation.} Recording the replicate-level selected indices resolves the comparison to individual physical factors. 
$\Sthree$ selects five physical factors at the following frequencies: $r_w$ 100\%, $H_u$ 93\%, $L$ 91\%, $H_l$ 90\%, and $K_w$ 43\%. The remaining three physical factors ($r$, $T_u$, $T_l$) have weak first order effects on $Q$. At this sample size, they are selected at approximately the same rate as the inactive factors ($1.5\%$). LassoCV retains the same three weak factors at $15$--$23\%$, indistinguishable from its own inactive background rate ($18$--$21\%$).

\FloatBarrier
\section{Discussion and Concluding Remarks}
\label{sec:disc}

In the supersaturated regime, $\Sthree$ screens to a far smaller selected set than the raw one-shot LassoCV output, so the Stage 2 response surface model stays small, and the selected set carries far fewer false positives. This advantage holds in every $(\text{case},p)$ aggregate cell and generally grows with $p$. Optimization quality is usually improved under the matched-$k_{\mathrm{m}}$ comparison.

The theoretical results are conditional by design.
Proposition \ref{prop:inactive_commitment} bounds the probability that a noise factor is ever fixed under Assumption \ref{as:noise}, and Appendix \ref{app:noise_diagnostics} reports round-level diagnostics for the constants that assumption involves. A quantitative link between the GQ schedule and standard high-dimensional regularity conditions, the irrepresentable \citep{zhao2006model} and restricted eigenvalue \citep{bickel2009simultaneous} conditions, remains open.

Stability selection \citep{meinshausen2010stability} at its default
stability threshold $\pi_{\mathrm{thr}} = 0.6$ under-selects (Type II $\approx 0.21$ and $0.61$ in Cases 1 and 2), leaving the Stage 2 response surface model with too few factors to recover the optimum, and the matched-budget comparison in Appendix \ref{app:ss} favors $\Sthree$ significantly in both
cases. The knockoff filter \citep{barber2015controlling} and its
Model-X extension \citep{candes2018panning} do not directly apply to binary two-level designs, where constructing valid knockoffs with fixed margins remains an open problem. Computational cost limits the use of the HILS two-stage construction of
\citet{stallrich2025optimal}, whose known-sign positive-correlation ideal motivates our design criterion. Its full sign-recovery procedure already costs roughly $80$ minutes per design at $n = 14$, $p = 20$. The smallest candidate count we simulate, $p = 40$, is twice that, so we do not include HILS in our comparisons. Controlled sequential bifurcation (CSB) and CSB-X are also related \citep{wan2006controlled,wan2010improving}, but they solve a different protocol. A fair comparison between CSB-X and a two-stage response surface modeling protocol would require its own budget-aware adaptation, which we leave to future work. Finally, two-level designs here give the cleanest discrete form of the theory and the natural screening corners, while continuous designs are left to follow-up work.


\if1\anon
\section{Acknowledgments}
\label{sec:ack}

The first author expresses sincere thanks for funding through the King's China Scholarship Council Scholarship Programme (File No. 202307720013). Kalliopi Mylona and Matteo Borrotti acknowledge support from COST Action CA21163 -- Text, functional and other high-dimensional data in econometrics: New models, methods, applications (HiTEc), supported by COST (European Cooperation in Science and Technology). The authors acknowledge the use of King's CREATE HPC. Retrieved July 14, 2026, from \url{https://doi.org/10.18742/rnvf-m076}.
\fi

\section{Data Availability Statement}

The computer code for generating and analyzing the simulated data is publicly available at \url{https://github.com/hansq358/S3}.

\bibliography{references}

@article{tibshirani1996regression,
  title={Regression shrinkage and selection via the lasso},
  author={Tibshirani, Robert},
  journal={Journal of the Royal Statistical Society Series B: Statistical Methodology},
  volume={58},
  number={1},
  pages={267--288},
  year={1996},
  publisher={Oxford University Press}
}

@article{booth1962some,
  title={Some systematic supersaturated designs},
  author={Booth, Kathleen HV and Cox, David R},
  journal={Technometrics},
  volume={4},
  number={4},
  pages={489--495},
  year={1962},
  publisher={Taylor \& Francis}
}

@article{fan2008sure,
  title={Sure independence screening for ultrahigh dimensional feature space},
  author={Fan, Jianqing and Lv, Jinchi},
  journal={Journal of the Royal Statistical Society Series B: Statistical Methodology},
  volume={70},
  number={5},
  pages={849--911},
  year={2008},
  publisher={Oxford University Press}
}

@inproceedings{jamieson2016non,
  title={Non-stochastic best arm identification and hyperparameter optimization},
  author={Jamieson, Kevin and Talwalkar, Ameet},
  booktitle={Artificial intelligence and statistics},
  pages={240--248},
  year={2016},
  organization={PMLR}
}

@article{szekely2007measuring,
  title={Measuring and testing dependence by correlation of distances},
  author={Sz{\'e}kely, G{\'a}bor J and Rizzo, Maria L and Bakirov, Nail K},
  journal={The Annals of Statistics},
  year={2007}
}

@article{li2012feature,
  title={Feature screening via distance correlation learning},
  author={Li, Runze and Zhong, Wei and Zhu, Liping},
  journal={Journal of the American Statistical Association},
  volume={107},
  number={499},
  pages={1129--1139},
  year={2012},
  publisher={Taylor \& Francis}
}

@techreport{harper1983sensitivity,
  title={Sensitivity/uncertainty analysis of a borehole scenario comparing Latin hypercube sampling and deterministic sensitivity approaches},
  author={Harper, William V and Gupta, Sumant K},
  year={1983},
  institution={Battelle Memorial Inst., Columbus, OH (USA). Office of Nuclear Waste Isolation}
}

@article{moon2012two,
  title={Two-stage sensitivity-based group screening in computer experiments},
  author={Moon, Hyejung and Dean, Angela M and Santner, Thomas J},
  journal={Technometrics},
  volume={54},
  number={4},
  pages={376--387},
  year={2012},
  publisher={Taylor \& Francis}
}

@article{singh2024factor,
  title={Factor selection in screening experiments by aggregation over random models},
  author={Singh, Rakhi and Stufken, John},
  journal={Computational Statistics \& Data Analysis},
  volume={194},
  pages={107940},
  year={2024},
  publisher={Elsevier}
}

@article{weese2021strategies,
  title={Strategies for supersaturated screening: Group orthogonal and constrained var (s) designs},
  author={Weese, Maria L and Stallrich, Jonathan W and Smucker, Byran J and Edwards, David J},
  journal={Technometrics},
  volume={63},
  number={4},
  pages={443--455},
  year={2021},
  publisher={Taylor \& Francis}
}

@article{singh2023selection,
  title={Selection of two-level supersaturated designs for main effects models},
  author={Singh, Rakhi and Stufken, John},
  journal={Technometrics},
  volume={65},
  number={1},
  pages={96--104},
  year={2023},
  publisher={Taylor \& Francis}
}

@article{stallrich2025optimal,
  title={An optimal design framework for lasso sign recovery},
  author={Stallrich, Jonathan W and Young, Kade and Weese, Maria L and Smucker, Byran J and Edwards, David J},
  journal={Journal of the Royal Statistical Society Series B: Statistical Methodology},
  volume={87},
  number={5},
  pages={1457--1480},
  year={2025},
  publisher={Oxford University Press UK}
}

@article{zou2005regularization,
  title={Regularization and variable selection via the elastic net},
  author={Zou, Hui and Hastie, Trevor},
  journal={Journal of the Royal Statistical Society Series B: Statistical Methodology},
  volume={67},
  number={2},
  pages={301--320},
  year={2005},
  publisher={Oxford University Press}
}

@article{zou2006adaptive,
  title={The adaptive lasso and its oracle properties},
  author={Zou, Hui},
  journal={Journal of the American statistical association},
  volume={101},
  number={476},
  pages={1418--1429},
  year={2006},
  publisher={Taylor \& Francis}
}

@article{fan2001variable,
  title={Variable selection via nonconcave penalized likelihood and its oracle properties},
  author={Fan, Jianqing and Li, Runze},
  journal={Journal of the American statistical Association},
  volume={96},
  number={456},
  pages={1348--1360},
  year={2001},
  publisher={Taylor \& Francis}
}

@article{meyer1995coordinate,
  title={The coordinate-exchange algorithm for constructing exact optimal experimental designs},
  author={Meyer, Ruth K and Nachtsheim, Christopher J},
  journal={Technometrics},
  volume={37},
  number={1},
  pages={60--69},
  year={1995},
  publisher={Taylor \& Francis}
}

@article{welch1974lower,
  title={Lower bounds on the maximum cross correlation of signals (corresp.)},
  author={Welch, Lloyd},
  journal={IEEE Transactions on Information theory},
  volume={20},
  number={3},
  pages={397--399},
  year={1974},
  publisher={IEEE}
}

@article{pan2020ball,
  title={Ball covariance: A generic measure of dependence in Banach space},
  author={Pan, Wenliang and Wang, Xueqin and Zhang, Heping and Zhu, Hongtu and Zhu, Jin},
  journal={Journal of the American Statistical Association},
  year={2020},
  publisher={Taylor \& Francis}
}

@article{lundberg2017unified,
  title={A unified approach to interpreting model predictions},
  author={Lundberg, Scott M and Lee, Su-In},
  journal={Advances in neural information processing systems},
  volume={30},
  year={2017}
}

@article{breiman2001random,
  title={Random forests},
  author={Breiman, Leo},
  journal={Machine learning},
  volume={45},
  number={1},
  pages={5--32},
  year={2001},
  publisher={Springer}
}

@article{phoa2014stepwise,
  title={The stepwise response refinement screener (SRRS)},
  author={Phoa, Frederick Kin Hing},
  journal={Statistica Sinica},
  volume={24},
  number={1},
  pages={197--210},
  year={2014},
  publisher={JSTOR}
}

@article{gutman2014augmenting,
  title={Augmenting supersaturated designs with Bayesian D-optimality},
  author={Gutman, Alex J and White, Edward D and Lin, Dennis KJ and Hill, Raymond R},
  journal={Computational Statistics \& Data Analysis},
  volume={71},
  pages={1147--1158},
  year={2014},
  publisher={Elsevier}
}

@article{candes2018panning,
  title={Panning for gold:‘model-X’knockoffs for high dimensional controlled variable selection},
  author={Candes, Emmanuel and Fan, Yingying and Janson, Lucas and Lv, Jinchi},
  journal={Journal of the Royal Statistical Society Series B: Statistical Methodology},
  volume={80},
  number={3},
  pages={551--577},
  year={2018},
  publisher={Oxford University Press}
}

@article{zhao2006model,
  title={On model selection consistency of Lasso},
  author={Zhao, Peng and Yu, Bin},
  journal={The Journal of Machine Learning Research},
  volume={7},
  pages={2541--2563},
  year={2006},
  publisher={JMLR. org}
}

@article{chen2008extended,
  title={Extended Bayesian information criteria for model selection with large model spaces},
  author={Chen, Jiahua and Chen, Zehua},
  journal={Biometrika},
  volume={95},
  number={3},
  pages={759--771},
  year={2008},
  publisher={Oxford University Press}
}

@article{bickel2009simultaneous,
  title={Simultaneous analysis of Lasso and Dantzig selector},
  author={Bickel, Peter J and Ritov, Ya’acov and Tsybakov, Alexandre B},
  journal={The Annals of Statistics},
  year={2009}
}

@article{meinshausen2010stability,
  title={Stability selection},
  author={Meinshausen, Nicolai and B{\"u}hlmann, Peter},
  journal={Journal of the Royal Statistical Society Series B: Statistical Methodology},
  volume={72},
  number={4},
  pages={417--473},
  year={2010},
  publisher={Oxford University Press}
}

@article{barber2015controlling,
  title={Controlling the false discovery rate via knockoffs},
  author={Barber, Rina Foygel and Cand{\`e}s, Emmanuel J},
  journal={The Annals of Statistics},
  year={2015}
}

@article{wan2006controlled,
  title={Controlled sequential bifurcation: A new factor-screening method for discrete-event simulation},
  author={Wan, Hong and Ankenman, Bruce E and Nelson, Barry L},
  journal={Operations Research},
  volume={54},
  number={4},
  pages={743--755},
  year={2006},
  publisher={INFORMS}
}

@article{wan2010improving,
  title={Improving the efficiency and efficacy of controlled sequential bifurcation for simulation factor screening},
  author={Wan, Hong and Ankenman, Bruce E and Nelson, Barry L},
  journal={INFORMS Journal on Computing},
  volume={22},
  number={3},
  pages={482--492},
  year={2010},
  publisher={INFORMS}
}

@article{byrd1995limited,
  title={A limited memory algorithm for bound constrained optimization},
  author={Byrd, Richard H. and Lu, Peihuang and Nocedal, Jorge and Zhu, Ciyou},
  journal={SIAM Journal on Scientific Computing},
  volume={16},
  number={5},
  pages={1190--1208},
  year={1995}
}

\phantomsection
\bigskip
\begin{center}
{\large\bf SUPPLEMENTARY MATERIAL}
\end{center}

\noindent
Supplementary material accompanying this article contains the technical proofs, supporting diagnostics, simulation details, baseline comparisons, comparisons of alternative design criteria and screening rules, benchmark summaries, and reproducibility files. The reproducibility bundle includes the Python code and scripts for generating the reported tables and figures.

\appendix

\section{Proofs and Supporting Results}
\label{app:proofs}

This appendix collects the technical support for the design and theory sections: the Welch-scale calibration result, the proof of Proposition \ref{prop:inactive_commitment} and its FDP corollaries, the Stage 1 run bound supporting Proposition \ref{prop:budget}, and round-level diagnostics for Assumption \ref{as:noise}, the positive-cone penalty, and late noise commitment.

\subsection{Welch-Scale Calibration of
\texorpdfstring{$\Var(s_+^{\mathrm{W}})$}{Var(s+W)}}
\label{app:welch_calibration}

\begin{proposition}[Welch-scale calibration of
\texorpdfstring{$\Var(s_+^{\mathrm{W}})$}{Var(s+W)}]
\label{prop:welch}
Let $X \in \{-1, +1\}^{n \times p}$ be a two-level design with
off-diagonal column inner products $s_{ij}$, and assume $p > n$.

\medskip\noindent
\textup{(a)} \textbf{Welch-scale normalization.} On the positive constant locus,
if $s_{ij} = \sqrt{n(p - n)/(p - 1)}$ for every $i \ne j$, then the
penalty term in \eqref{eq:welch_crit} satisfies
\begin{equation}
  \eta_{\mathrm{W}}(n, p) \, (\bar s)^2 \;=\; 1.
  \label{eq:welch_calibration}
\end{equation}

\medskip\noindent
\textup{(b)} \textbf{Welch-scale lower bound and its non-attainment.}
Restricted to the constant locus $\{s_{ij} = s_0 : i \ne j\}$ with
$s_0 \ge 0$, the criterion in \eqref{eq:welch_crit} reduces to
$\eta_{\mathrm{W}}\, s_0^2$. The Welch lower bound \citep{welch1974lower}
gives $\max_{i \ne j} s_{ij}^2 \ge n(p-n)/(p-1)$, so on this locus the
criterion is at least $1$, with equality only if every off-diagonal
equals $s^\star = \sqrt{n(p-n)/(p-1)}$. This configuration is not
attainable by any two-level design when $p > n$ and $n \ge 2$. If all off-diagonals of
$G = X^\top X$ equalled a common $s_0$ with $0 < s_0 < n$, then
$G = (n - s_0) I_p + s_0 \mathbf{1}\mathbf{1}^\top$ would have full
rank $p$, whereas $\rank(X^\top X) \le n < p$. The Welch
magnitude therefore sets a scale for the unavoidable aliasing rather
than an attainable target, and $\Var(s_+^{\mathrm{W}})$ is a
Welch-scaled positive-cone criterion.

\medskip\noindent
\textup{(c)} \textbf{Orthogonal limit.} When $p \le n$,
$\eta_{\mathrm{W}}(n, p) = 1$ and, on the positive cone, the criterion
reduces to the classical $E(s^2) = \overline{s^2}$ near-orthogonality
objective \citep{booth1962some}, whose minimum value $0$ at
$s_{ij} = 0$ for all $i \ne j$ is achieved when an orthogonal two-level
design exists.
\end{proposition}

\begin{proof}[Proof of Proposition \ref{prop:welch}]
\textup{(a)} Under the stated condition,
$|s_{ij}| = \sqrt{n(p - n)/(p - 1)}$ for every $i \ne j$. With the
positive-cone constraint $s_{ij} \ge 0$, this gives
$s_{ij} = +\sqrt{n(p - n)/(p - 1)}$ uniformly. Hence
\[
  \bar s
  \;=\;
  \frac{1}{\binom{p}{2}} \sum_{i < j} s_{ij}
  \;=\;
  \sqrt{\frac{n(p - n)}{p - 1}},
\]
and therefore
$(\bar s)^2 = n(p - n)/(p - 1)$. Multiplying by
$\eta_{\mathrm{W}}(n, p) = (p - 1)/[n(p - n)]$ gives
\[
  \eta_{\mathrm{W}} \, (\bar s)^2
  \;=\;
  \frac{p - 1}{n(p - n)} \cdot \frac{n(p - n)}{p - 1}
  \;=\; 1,
\]
which is \eqref{eq:welch_calibration}.

\textup{(b)} On the constant locus $\{s_{ij} = s_0 : i \ne j\}$ with
$s_0 \ge 0$, the variance term of \eqref{eq:welch_crit} vanishes and the
positive-cone penalty is zero, so the criterion equals
$\eta_{\mathrm{W}}\, s_0^2$. The Welch lower bound \citep{welch1974lower}
gives $\max_{i \ne j} s_{ij}^2 \ge n(p - n)/(p - 1)$ for $p > n$, hence
$s_0 \ge s^\star = \sqrt{n(p - n)/(p - 1)}$ and
$\eta_{\mathrm{W}}\, s_0^2 \ge \eta_{\mathrm{W}}\,(s^\star)^2 = 1$ by
part (a). For the non-attainment claim, suppose every off-diagonal of
$G = X^\top X$ equals a common $s_0$ with $0 < s_0 < n$. Then
$G = (n - s_0) I_p + s_0 \mathbf{1}\mathbf{1}^\top$, whose eigenvalues
are $n + (p - 1)s_0 > 0$ once and $n - s_0 > 0$ with multiplicity
$p - 1$, so $\rank(G) = p$, while
$\rank(X^\top X) = \rank(X) \le n < p$, a
contradiction. Since $s^\star \in (0, n)$ whenever $n \ge 2$, the
Welch-saturating positive constant locus is empty for any real design
with $p > n$. Coordinate exchange therefore drives the design toward the
positive cone and the Welch scale only as far as the discrete geometry
allows, and Table \ref{tab:negpair_diag} reports that the cone is
reached in $89$--$98\%$ of recovered designs.

\textup{(c)} For $p \le n$, $\eta_{\mathrm{W}}(n, p) = 1$ and the
criterion reduces to $\overline{s^2} - (\bar s)^2 + (\bar s)^2 = \overline{s^2}$.
When an orthogonal two-level design exists it satisfies $s_{ij} = 0$ for
all $i \ne j$, which zeros each of the three terms and achieves the
minimum value $0$, with the positive-cone constraint holding trivially
at equality.
\end{proof}

\subsection{Proof of Proposition \ref{prop:inactive_commitment} and Its Corollaries}

\begin{proof}[Proof of Proposition \ref{prop:inactive_commitment}]
For $j \in \cN$ to enter the final fixed set $\cF$, there is a
first round $T_j \in \{1, \ldots, S_{\max}\}$ at which $j$ is
committed by the elbow rule, and the events $C_t(j)$ of
Section \ref{sec:theory} are empty for rounds that are not executed. If $T_j = t$, then round $t$ was executed, $j$ was still open at
round $t$, and $j$ was committed there. Remaining open through rounds
$1, \ldots, t-1$ requires exceeding each earlier elimination
threshold, because an open factor whose score falls at or below
$\tau_s$ and which is not committed at round $s$ is eliminated at
round $s$. Hence, with $H_t(j)$ and $C_t(j)$ defined in
Section \ref{sec:theory},
\[
  \{T_j = t\}
  \;\subseteq\;
  \biggl(\bigcap_{s<t} H_s(j)\biggr)
  \cap \{j \in \cO_t\} \cap C_t(j).
\]
For each fixed $t$, define
\[
  E_m^{(t)}
  \;=\;
  \bigcap_{s=1}^{m}
    H_s(j),
  \qquad m = 0, 1, \ldots, t-1,
\]
with $E_0^{(t)} = \Omega$. The event
$E_{m-1}^{(t)} \cap \{j \in \cO_m\}$ is contained in
$\{j \in \cO_m\}$ and is $\cG_{m-1}$-measurable, so the survival part
of Assumption \ref{as:noise} gives
\[
  \bbP\bigl(E_m^{(t)}\bigr)
  \;=\; \bbP\bigl(H_m(j) \,\bigm|\,
    E_{m-1}^{(t)} \cap \{j \in \cO_m\}\bigr)
  \cdot \bbP\bigl(E_{m-1}^{(t)} \cap \{j \in \cO_m\}\bigr)
  \;\le\; \rho_m\,\bbP\bigl(E_{m-1}^{(t)}\bigr).
\]
When the conditioning event is null the same bound holds directly, since $E_m^{(t)} \subseteq E_{m-1}^{(t)} \cap \{j \in \cO_m\}$.
Induction on $m$ with $\bbP(E_0^{(t)}) = 1$ yields
$\bbP(E_{t-1}^{(t)}) \le \prod_{s<t}\rho_s$. The event
$E_{t-1}^{(t)} \cap \{j \in \cO_t\}$ is $\cG_{t-1}$-measurable and is
contained in $\{j \in \cO_t\}$, so the commitment part of
Assumption \ref{as:noise} gives
\[
  \bbP(T_j = t)
  \;\le\;
  \bbP\bigl(C_t(j) \bigm| E_{t-1}^{(t)} \cap \{j \in \cO_t\}\bigr)
  \cdot \bbP\bigl(E_{t-1}^{(t)} \cap \{j \in \cO_t\}\bigr)
  \;\le\; \delta_t\prod_{s<t}\rho_s.
\]
Summing over the possible first commitment rounds gives
\[
  \bbP(j \in \cF)
  \;=\; \sum_{t=1}^{S_{\max}}\bbP(T_j = t)
  \;\le\; \sum_{t=1}^{S_{\max}}\delta_t\prod_{s<t}\rho_s,
\]
which proves the single-factor bound of Proposition \ref{prop:inactive_commitment}.
Summing over the $p - k$ noise factors yields \eqref{eq:efp}.
\end{proof}

\begin{corollary}[Observed-size FDP bound]
\label{cor:fdp}
Define
\[
  B = \sum_{t=1}^{S_{\max}}\delta_t\prod_{s<t}\rho_s.
\]
Under Assumption \ref{as:noise}, for any deterministic $m \ge 1$,
\begin{equation}
  \bbE\bigl[\mathrm{FDP}\,\mathbf 1\{|\cF| \ge m\}\bigr]
  \;\le\;
  \frac{(p - k)\,B}{m}.
  \label{eq:efdp}
\end{equation}
This expectation bound depends on the survival and commitment
constants through $B$ and on the deterministic size threshold $m$.
\end{corollary}

\begin{proof}[Proof of Corollary \ref{cor:fdp}]
On the event $\{|\cF| \ge m\}$,
\[
  \mathrm{FDP}
  \;=\;
  \frac{|\cF \cap \cN|}{|\cF|\vee 1}
  \;\le\;
  \frac{|\cF \cap \cN|}{m}.
\]
Multiplying by the event indicator, taking expectations, and applying
\eqref{eq:efp} gives \eqref{eq:efdp}.
\end{proof}

\begin{corollary}[Clean-ladder special case]
\label{cor:clean_ladder}
If $\delta_t = 0$ for every $t < s_0$, then
\begin{equation}
  \bbE\bigl[|\cF \cap \cN|\bigr]
  \;\le\;
  (p-k)\sum_{t=s_0}^{S_{\max}}\delta_t\prod_{s<t}\rho_s.
  \label{eq:clean_ladder}
\end{equation}
Thus the first possible inactive commitment inherits the survival
product through the clean rounds $1,\ldots,s_0-1$.
\end{corollary}

\subsection{Stage 1 Run Bound}

\begin{proposition}[GQ elimination rate and Stage 1 run count]
\label{prop:elim}
Consider GQ screening with parameters $c, q \in (0, 1)$, the default
LassoCV scoring, and linear-interpolation sample quantiles, starting
from $|\cO_1| = |\cC_1| = p \ge 2$ open candidates, and define
$q_{\min} = q \cdot \min(1, c)$.

\medskip\noindent
\textup{(a)} Every executed round satisfies
$|\cO_{s+1}| \le (1 - q_{\min})\,|\cO_s|$, so
$|\cO_s| \le p (1 - q_{\min})^{s-1}$.

\medskip\noindent
\textup{(b)} After $S$ rounds the open set satisfies $|\cO_{S+1}| \le k$
whenever $S \ge \log_{1/(1-q_{\min})}(p/k)$.

\medskip\noindent
\textup{(c)} The total Stage 1 runs over $S$ executed rounds satisfy
\begin{equation}
  N_{\mathrm{S1}}
  \;\le\;
  \frac{c \, p \bigl[1 - (1 - q_{\min})^{S}\bigr]}{q_{\min}} + 2S
  \;\le\;
  \frac{p}{q} + 2S.
  \label{eq:N1bound}
\end{equation}
\end{proposition}

\begin{proof}[Proof of Proposition \ref{prop:elim}]
\textup{(a)} The cumulative run count increases, while the
candidate set does not grow. Thus, in every executed round,
$N_s \ge N_1=\max(2,\lceil cp\rceil)\ge cp$ and
$|\cC_s|\le p$. Hence $N_s/|\cC_s|\ge c$, and
$q_s=q\min(1,N_s/|\cC_s|)\ge q\min(1,c)=q_{\min}$.

Write $m=|\cO_s|$ and first suppose $m\ge2$. Order the open scores as $I_{(1)}\le\cdots\le I_{(m)}$,
and set $h=(m-1)q_s$. The threshold $\tau_s$ is obtained by linear interpolation between $I_{(\lfloor h\rfloor+1)}$ and $I_{(\lfloor h\rfloor+2)}$. Thus $\tau_s\ge I_{(\lfloor h\rfloor+1)}$, so at least
$\lfloor h\rfloor+1$ open scores lie at or below $\tau_s$. These factors leave the open set by either elimination or commitment. If some open score exceeds $\tau_s$, the largest open score is positive, so its sign is recorded and the elbow prefix commits at least the top-scoring open factor, which also leaves the open set. Then, using
$\lfloor h \rfloor + 1 \ge h$,
\[
  |\cO_{s+1}|
  \;\le\; m - (\lfloor h \rfloor + 1) - 1
  \;\le\; m - (m - 1)\, q_s - 1
  \;=\; (m - 1)(1 - q_s)
  \;\le\; (1 - q_{\min})\, m.
\]
If no open score exceeds $\tau_s$, every open factor is at or below the
threshold, so all of them leave the open set and $|\cO_{s+1}| = 0$.
If $m=1$, this is a later round because $p\ge2$. Its open factor
survived the preceding threshold without being fixed, so it had a
positive score and its direction was recorded. The remaining open factor is therefore committed, leaving the open set empty. Iterating over executed
rounds yields $|\cO_s| \le p (1 - q_{\min})^{s-1}$.

\noindent
\textup{(b)} By part (a), $|\cO_{S+1}| \le p (1 - q_{\min})^{S} \le k$
once $S \ge \log_{1/(1-q_{\min})}(p/k)$.

\noindent
\textup{(c)} The total number of runs is
$N_{\mathrm{S1}}=\sum_{s=1}^{S}n_s$.
The allocation rule and part (a) give
\[
  n_s=\max(2,\lceil c|\cO_s|\rceil)
  \le c|\cO_s|+2
  \le cp(1-q_{\min})^{s-1}+2.
\]
Summing over $s=1,\ldots,S$ gives
\[
  N_{\mathrm{S1}}
  \;\le\;
  c \, p \sum_{s=0}^{S-1} (1 - q_{\min})^{s} + 2S
  \;=\;
  \frac{c \, p \bigl[1 - (1 - q_{\min})^{S}\bigr]}{q_{\min}} + 2S.
\]
Since $1 - (1 - q_{\min})^{S} \le 1$ and $q_{\min} = qc$ for
$c \in (0, 1)$, the simpler bound
$N_{\mathrm{S1}} \le c\,p/q_{\min} + 2S = p/q + 2S$ follows.
\end{proof}

Part (a) bounds the number of open factors by a geometric sequence, similar to the reduction in candidate numbers in successive halving \citep{jamieson2016non}. The two methods collect data differently, so the guarantees for successive halving do not directly apply here.

\subsection{Round-Level Diagnostics for Assumption \ref{as:noise}}\label{app:noise_diagnostics}

To examine the constants of Assumption \ref{as:noise} in the operating regime of the procedure, we record the survival and commitment counts directly during Stage 1 execution on the $\Sthree$ arm of the check under a fixed total budget (Appendix \ref{app:fixed_budget}) at the default split
$N_{\mathrm{S1}}/N = 0.5$, with $R = 500$ replications per cell. The round-level survival rates
agree with the matched comparison of Section \ref{sec:main_comparison}. At each round $t$, the log
records the open inactive count $|\cN \cap \cO_t|$, the number of open
inactive factors whose score exceeds the elimination threshold, and the
number committed by the elbow rule. The empirical diagnostics are
\begin{equation}
  \hat\rho_t =
  \frac{\textstyle \sum_{r = 1}^{R}
        |\{j \in \cN \cap \cO_t^{(r)} : I_j^{(t, r)} > \tau_t^{(r)}\}|}
       {\textstyle \sum_{r = 1}^{R} |\cN \cap \cO_t^{(r)}|},
  \qquad
  \hat\delta_t =
  \frac{\textstyle \sum_{r = 1}^{R}
        |\{j \in \cN \cap \cO_t^{(r)} : C_t^{(r)}(j)\}|}
       {\textstyle \sum_{r = 1}^{R} |\cN \cap \cO_t^{(r)}|}.
  \label{eq:rho_delta_direct}
\end{equation}
The commitment count includes factors fixed at or below the threshold,
so $\hat\delta_t$ need not be smaller than $\hat\rho_t$.
These rates summarize average survival and commitment across
the simulated histories and provide empirical diagnostics
for Assumption \ref{as:noise}.

\begin{table}[H]
\centering
\caption{Pooled round-level survival and commitment diagnostics for the
$\Sthree$ Stage 1 arm at the default split, with $R = 500$ replications
per cell.
The $\hat\rho_{\ge 3}$ and $\hat\delta_{\ge 3}$ columns pool every round from the third onward, where only a few inactive factors remain open.
The column $s_0$ is set by the prespecified rule: the largest
$s_0 \in \{1, \ldots, 4\}$ with $\hat\delta_t < 0.01$ for every
$t < s_0$.}
\label{tab:noise_diagnostics}
\small
\begin{tabular}{c c rrr rrr c}
\toprule
Case & $p$ & $\hat\rho_1$ & $\hat\rho_2$ & $\hat\rho_{\ge 3}$ &
$\hat\delta_1$ & $\hat\delta_2$ & $\hat\delta_{\ge 3}$ & $s_0$ \\
\midrule
1 & 40  & $0.278$ & $0.225$ & $0.214$ & $0.0042$ & $0.079$ & $0.319$ & $2$ \\
1 & 80  & $0.196$ & $0.249$ & $0.248$ & $0.0002$ & $0.038$ & $0.290$ & $2$ \\
1 & 120 & $0.137$ & $0.253$ & $0.246$ & $<0.0001$ & $0.036$ & $0.268$ & $2$ \\
1 & 200 & $0.102$ & $0.259$ & $0.245$ & $<0.0001$ & $0.017$ & $0.216$ & $2$ \\
\midrule
2 & 40  & $0.227$ & $0.285$ & $0.343$ & $0.0177$ & $0.152$ & $0.480$ & $1$ \\
2 & 80  & $0.157$ & $0.297$ & $0.347$ & $0.0059$ & $0.134$ & $0.456$ & $2$ \\
2 & 120 & $0.123$ & $0.310$ & $0.359$ & $0.0043$ & $0.101$ & $0.457$ & $2$ \\
2 & 200 & $0.084$ & $0.305$ & $0.333$ & $0.0006$ & $0.095$ & $0.448$ & $2$ \\
\bottomrule
\end{tabular}
\end{table}

The pooled first-round commitment rate is
$915 / 420000 = 0.00218$, so the prespecified pooled rule selects
$s_0 = 2$. The same pooled decision holds within each case. Seven of
the eight case-by-$p$ summaries also select $s_0 = 2$. The exception is
the weak-signal small-$p$ cell, Case 2 at $p = 40$, where
$\hat\delta_1 = 0.0177$. The second round is not clean: the pooled
$\hat\delta_2$ is $0.0717$, and the case-by-$p$ values range from
$0.017$ to $0.152$. The observed first-round commitments support retaining
the general $\delta_t$ terms in the bound.

Extending the diagnostics past the second round, the pooled survival
rate $\hat\rho_{\ge3}$ stays between $0.21$ and $0.36$, while the
commitment rate $\hat\delta_{\ge3}$ rises to between $0.22$ and $0.48$
as the open inactive set shrinks to a few factors per replicate.
The sum $\hat B=\sum_t\hat\delta_t\prod_{s<t}\hat\rho_s$
over the recorded rounds gives a committed-count bound
$(p-k)\hat B$ of $2.0$ to $2.3$ in Case 1 and $3.7$ to $5.4$
in Case 2, which exceeds the observed mean $|\cF\cap\cN|$
($1.5$ to $1.6$ in Case 1, $2.4$ to $3.5$ in Case 2)
in every cell. Both $\hat B$ and the observed counts are calculated
from the same logs, so this comparison checks their internal
consistency. The rising $\hat\delta_t$ is offset by the geometrically
contracting survival weight $\prod_{s<t}\hat\rho_s$, so each round
past the second adds less than $0.04$ to $\hat B$.

\subsection{Empirical Validation of the Positive-Cone Penalty}
\label{app:cone_validation}

The positive-cone enforcement in \eqref{eq:welch_crit} relies on the
penalty weight $\lambda = 10^6$ acting as an effectively hard
constraint. To verify this we record, for every recovered Stage 1
design across the simulation grid, the number of negative off-diagonal
pairs $|\{(i, j): i < j, s_{ij} < 0\}|$ and the minimum off-diagonal
inner product $\min_{i < j} s_{ij}$.

\begin{table}[H]
\centering
\caption{Positive-cone diagnostics on the
final selected columns of $\Var(s_+^{\mathrm{W}})$ recovered
designs at the end of Stage 1, aggregated across all $63$ $(c, q)$
cells $\times$ $R = 500$ replications per cell. ``frac zero'' is the
fraction of replications with no negative pair in the recovered
selected submatrix. Replications selecting fewer than two factors have
no off-diagonal pair and are excluded.}
\label{tab:negpair_diag}
\small
\begin{tabular}{c c r rr c rr}
\toprule
Case & $p$ & $n$ designs
  & mean neg-pair & max neg-pair
  & frac zero
  & mean $\min s$ & worst $\min s$ \\
\midrule
1 & 40  & $31{,}464$ & $0.122$ & $19$ & $94.9\%$ & $+0.86$ & $-17$ \\
1 & 80  & $31{,}500$ & $0.113$ & $24$ & $94.4\%$ & $+2.04$ & $-21$ \\
1 & 120 & $31{,}500$ & $0.070$ & $16$ & $96.3\%$ & $+4.05$ & $-24$ \\
1 & 200 & $31{,}500$ & $0.025$ & $16$ & $98.5\%$ & $+8.44$ & $-24$ \\
\midrule
2 & 40  & $29{,}700$ & $0.317$ & $44$ & $89.3\%$ & $+0.99$ & $-25$ \\
2 & 80  & $31{,}320$ & $0.245$ & $35$ & $90.5\%$ & $+2.12$ & $-58$ \\
2 & 120 & $31{,}500$ & $0.132$ & $48$ & $94.5\%$ & $+4.24$ & $-63$ \\
2 & 200 & $31{,}500$ & $0.070$ & $41$ & $96.6\%$ & $+8.33$ & $-55$ \\
\bottomrule
\end{tabular}
\end{table}

Across $\approx 250{,}000$ recovered designs the mean negative-pair
count among the selected columns is at most $0.32$ and the fraction of selected submatrices in the positive cone ranges from $89.3\%$
(Case 2, $p = 40$) to $98.5\%$ (Case 1, $p = 200$). The diagnostic
is restricted to the selected columns because those are the design
columns that actually enter Stage 2 modeling. The cumulative
Stage 1 design contains many additional columns that were
subsequently eliminated and whose pairwise inner products are not
relevant downstream. The empirical exception of a small boundary
cohort sits in the weak-signal Case 2 at low $p$, where coordinate
exchange terminates before all pairs cross into the positive cone.

\subsection{Mis-fixation and Late Correction}
\label{app:misfix}

Locking committed factors at $\pm 1$ preserves their contribution to
interactions in subsequent rounds, but in principle propagates errors if a
factor is fixed at the wrong sign or a noise factor is committed.
Empirically, at the representative cell $(c, q) = (0.50, 0.65)$ and
$k_{\mathrm{true}} = 5$, the fraction of replications with zero noise in
the final fixed set $\cF$ is $30$--$44\%$ for Case 1 and $8$--$14\%$ for
Case 2 across $p \in \{40, 80, 120, 200\}$, with mean $|\cF \cap \cN|$
of $1.2$--$1.5$ (Case 1) and $2.6$--$3.4$ (Case 2), and most of this
commitment occurs after the open inactive set has already contracted
(Table \ref{tab:noise_diagnostics}). A possible late correction is to recompute $\sign(\hat\beta_j)$
from the terminal cumulative fit. This would change signs without
adding runs or changing the selected set, but has not been evaluated here.

\section{Simulation Details and Metrics Accounting for Class Imbalance}
\label{app:simulation_details}

\subsection{Test Case Coefficients}

\begin{table}[H]
\centering
\caption{True coefficients for the two main cases at $k_{\mathrm{true}} = 5$
(intercept $\beta_0 = 20$, $\sigma = 1.0$, $\sigma_b = 1.0$). The interaction
and quadratic structure is shared between the two cases, and only the
main-effect coefficient vector differs. For $k_{\mathrm{true}} \ne 5$,
main-effect coefficients are taken from the first $k_{\mathrm{true}}$
entries, and the number of interactions and quadratics grows with
$k_{\mathrm{true}}$ (capped at $5$ interactions and $4$ quadratics).}
\label{tab:cases}
\small
\begin{tabular}{l *{5}{c} @{\hspace{1em}} *{2}{c} @{\hspace{1em}} *{2}{c}
  @{\hspace{1em}} c}
\toprule
 & \multicolumn{5}{c}{Main effects}
 & \multicolumn{2}{c}{Interactions}
 & \multicolumn{2}{c}{Quadratics} & \\
\cmidrule(lr){2-6}\cmidrule(lr){7-8}\cmidrule(lr){9-10}
Case
 & $x_0$ & $x_1$ & $x_2$ & $x_3$ & $x_4$
 & $x_0 x_1$ & $x_1 x_2$
 & $x_0^2$ & $x_3^2$
 & Regime \\
\midrule
Case 1
 & $4.0$ & $3.0$ & $2.0$ & $1.5$ & $1.0$
 & $2.0$ & $-1.5$
 & $-1.5$ & $1.0$
 & Main-effect dominated \\
Case 2
 & $2.0$ & $1.5$ & $1.0$ & $0.5$ & $0.5$
 & $2.0$ & $-1.5$
 & $-1.5$ & $1.0$
 & Weak signal \\
\bottomrule
\end{tabular}
\end{table}

Table \ref{tab:cases} fixes the $k_{\mathrm{true}} = 5$ coefficients.
For the $k$-sweep of Table \ref{tab:ksweep}, the main-effect vector is
the length-$k_{\mathrm{true}}$ prefix of
$[4.0, 3.0, 2.0, 1.5, 1.0, 0.8, 0.5, 0.3, 0.2, 0.1, 0.05, 0.03]$
(Case 1) and
$[2.0, 1.5, 1.0, 0.5, 0.5, 0.3, 0.2, 0.15, 0.1, 0.05, 0.03, 0.02]$
(Case 2), and the interaction and quadratic terms grow with
$k_{\mathrm{true}}$ as listed in Table \ref{tab:case_ksweep}.

\begin{table}[H]
\centering
\caption{Interaction and quadratic coefficients activated at each
$k_{\mathrm{true}}$ used in the $k$-sweep (Table \ref{tab:ksweep}),
shared by Cases 1 and 2. Interactions sit on the consecutive pairs
$(x_0, x_1), (x_1, x_2), \ldots$ in order and quadratics on indices
$\{0, 3, 6, 9\}$ in order, so the $j$-th listed interaction
coefficient multiplies $x_{j-1} x_j$ and the $j$-th quadratic
coefficient multiplies $x_{3(j-1)}^2$.}
\label{tab:case_ksweep}
\small
\begin{tabular}{c l l}
\toprule
$k_{\mathrm{true}}$ & Interaction coefficients & Quadratic coefficients \\
\midrule
$4$  & $2.0,\ -1.5$                    & $-1.5,\ 1.0$ \\
$5$  & $2.0,\ -1.5$                    & $-1.5,\ 1.0$ \\
$6$  & $2.0,\ -1.5,\ 1.2$              & $-1.5,\ 1.0$ \\
$7$  & $2.0,\ -1.5,\ 1.2$              & $-1.5,\ 1.0,\ -0.8$ \\
$8$  & $2.0,\ -1.5,\ 1.2,\ -1.0$       & $-1.5,\ 1.0,\ -0.8$ \\
$9$  & $2.0,\ -1.5,\ 1.2,\ -1.0$       & $-1.5,\ 1.0,\ -0.8$ \\
$10$ & $2.0,\ -1.5,\ 1.2,\ -1.0,\ 0.8$ & $-1.5,\ 1.0,\ -0.8,\ 0.6$ \\
$11$ & $2.0,\ -1.5,\ 1.2,\ -1.0,\ 0.8$ & $-1.5,\ 1.0,\ -0.8,\ 0.6$ \\
$12$ & $2.0,\ -1.5,\ 1.2,\ -1.0,\ 0.8$ & $-1.5,\ 1.0,\ -0.8,\ 0.6$ \\
\bottomrule
\end{tabular}
\end{table}

\subsection{Metrics Accounting for Class Imbalance}

Type I and Type II errors are the proportions of inactive factors
incorrectly selected and of active factors missed, respectively. At
cell-mean level we compute confusion matrix entries
$\mathrm{TP} = k (1 - \mathrm{TypeII})$,
$\mathrm{FP} = (p - k) \cdot \mathrm{TypeI}$,
$\mathrm{FN} = k \cdot \mathrm{TypeII}$,
$\mathrm{TN} = (p - k) (1 - \mathrm{TypeI})$, and define
\begin{equation}
  F_1 = \frac{2\,\mathrm{TP}}{2\,\mathrm{TP} + \mathrm{FP} + \mathrm{FN}},
  \qquad
  \mathrm{MCC} = \frac{\mathrm{TP}\,\mathrm{TN}
    - \mathrm{FP}\,\mathrm{FN}}
    {\sqrt{(\mathrm{TP}+\mathrm{FP})(\mathrm{TP}+\mathrm{FN})
    (\mathrm{TN}+\mathrm{FP})(\mathrm{TN}+\mathrm{FN})}}.
  \label{eq:f1mcc}
\end{equation}
$F_1$ and MCC depend on both error rates and the relative numbers of active and inactive factors.
\section{Baseline Screening Methods}
\label{app:baselines}

This appendix justifies the external baseline choices used in the main text: LassoCV as a standard one-shot comparator, the screening behavior of five alternative one-shot rules across the
$n/p$ regime, a comparison against stability selection at matched budgets, and stricter penalties on the same Lasso path. All baselines are evaluated inside the same two-level design framework as the main simulation study, with $R = 500$ replications per row.

\subsection{One-shot LassoCV versus OLS-top-15}
\label{app:lassocv}

One-shot LassoCV is the natural comparator in the
supersaturated regime because it implicitly performs variable selection
through coefficient sparsity. Ordinary least squares performs no
such selection, so we truncate it to the top-$15$ factors ranked by
absolute coefficient. The truncation level of 15 factors matches the cap used
for Stage 2. At $n \le p$, OLS is statistically unstable. We evaluate both methods on Case 1
($p = 80$, $k_{\mathrm{true}} = 5$, $\Var(s_+^{\mathrm{W}})$
design) across nine $n/p$ ratios densely sampled around the
crossover $n = p$.

\begin{table}[H]
\centering
\caption{LassoCV versus OLS top-15 on two-level designs, Case 1,
$p = 80$, $k_{\mathrm{true}} = 5$, $R = 500$ replications.
Each row reports cell means.}
\label{tab:lassocv_vs_ols}
\small
\begin{tabular}{c c rrr rrr}
\toprule
$n/p$ & $n$
  & La $k_{\mathrm{sel}}$ & La T1 & La T2
  & OLS $k_{\mathrm{sel}}$ & OLS T1 & OLS T2 \\
\midrule
0.625 & 50  & 19.71 & 0.198 & 0.034 & 15 & 0.144 & 0.167 \\
0.750 & 60  & 19.56 & 0.195 & 0.017 & 15 & 0.143 & 0.147 \\
0.875 & 70  & 19.26 & 0.191 & 0.009 & 15 & 0.147 & 0.200 \\
1.000 & 80  & 17.73 & 0.170 & 0.010 & 15 & 0.173 & 0.596 \\
1.125 & 90  & 17.95 & 0.173 & 0.008 & 15 & 0.142 & 0.127 \\
1.250 & 100 & 17.79 & 0.171 & 0.004 & 15 & 0.139 & 0.084 \\
1.500 & 120 & 17.46 & 0.166 & 0.002 & 15 & 0.134 & 0.017 \\
1.750 & 140 & 17.64 & 0.169 & 0.001 & 15 & 0.134 & 0.006 \\
2.000 & 160 & 17.48 & 0.166 & 0.000 & 15 & 0.133 & 0.001 \\
\bottomrule
\end{tabular}
\end{table}

OLS top-15 is markedly unstable near the boundary $n = p$.
At $n/p = 1.0$, the Type II error jumps to $0.596$, almost three
times its value at $n/p = 0.875$ and nearly five times its value at
$n/p = 1.125$. The pattern reflects how unstable the OLS coefficient estimates become when the design is near rank-deficient.

\subsection{One-shot LassoCV versus five alternative one-shot rules}
\label{app:onealt}

To check that LassoCV is a competitive default one-shot comparator,
Table \ref{tab:six_alt} compares it with five alternative one-shot
rules at $n/p\in\{0.625,1.0,1.5\}$ in Cases 1 and 2.

\begin{table}[H]
\centering
\caption{Six one-shot baselines on two-level designs, $p = 80$,
$k_{\mathrm{true}} = 5$, $R = 500$, at three sample-size ratios
$n/p$. Each row is one method. Entries report the cell-mean selected
size $k_{\mathrm{sel}}$, the Type I (T1) and Type II (T2) errors, and
the end-to-end Stage 2 LOSS at each method's own selection (capped at
the Stage 2 factor limit of 15).}
\label{tab:six_alt}
\footnotesize
\setlength{\tabcolsep}{3.4pt}
\begin{tabular}{l rrrr rrrr rrrr}
\toprule
 & \multicolumn{4}{c}{$n/p = 0.625$}
 & \multicolumn{4}{c}{$n/p = 1.0$}
 & \multicolumn{4}{c}{$n/p = 1.5$} \\
\cmidrule(lr){2-5}\cmidrule(lr){6-9}\cmidrule(lr){10-13}
Method & $k_{\mathrm{sel}}$ & T1 & T2 & LOSS
       & $k_{\mathrm{sel}}$ & T1 & T2 & LOSS
       & $k_{\mathrm{sel}}$ & T1 & T2 & LOSS \\
\midrule
\multicolumn{13}{l}{\textbf{Case 1 (main-effect dominated)}}\\
LassoCV         & 19.5 & 0.20 & 0.03 & 0.33 & 18.4 & 0.18 & 0.01 & 0.14 & 17.7 & 0.17 & 0.00 & 0.09 \\
ElasticNetCV    & 20.1 & 0.20 & 0.03 & 0.34 & 18.4 & 0.18 & 0.01 & 0.14 & 17.7 & 0.17 & 0.00 & 0.08 \\
Adaptive Lasso  & 23.2 & 0.25 & 0.05 & 0.55 & 25.6 & 0.28 & 0.02 & 0.23 & 21.9 & 0.23 & 0.01 & 0.11 \\
SCAD            & 29.0 & 0.32 & 0.03 & 0.41 & 38.6 & 0.45 & 0.01 & 0.25 & 46.1 & 0.55 & 0.00 & 0.11 \\
SIS             & 12.0 & 0.10 & 0.14 & 1.09 & 15.0 & 0.14 & 0.09 & 0.62 & 15.0 & 0.14 & 0.04 & 0.26 \\
OLS top-15      & 15.0 & 0.14 & 0.16 & 1.25 & 15.0 & 0.17 & 0.62 & 6.49 & 15.0 & 0.13 & 0.02 & 0.18 \\
\midrule
\multicolumn{13}{l}{\textbf{Case 2 (weak signal)}}\\
LassoCV         & 16.6 & 0.17 & 0.26 & 1.69 & 15.3 & 0.15 & 0.20 & 1.33 & 14.9 & 0.14 & 0.15 & 1.01 \\
ElasticNetCV    & 20.7 & 0.22 & 0.23 & 1.67 & 16.6 & 0.17 & 0.19 & 1.34 & 15.2 & 0.15 & 0.15 & 1.05 \\
Adaptive Lasso  & 27.7 & 0.32 & 0.21 & 1.76 & 24.7 & 0.27 & 0.17 & 1.42 & 21.4 & 0.23 & 0.13 & 1.11 \\
SCAD            & 25.6 & 0.29 & 0.22 & 1.76 & 33.6 & 0.39 & 0.12 & 1.39 & 41.3 & 0.49 & 0.07 & 1.12 \\
SIS             & 12.0 & 0.11 & 0.30 & 1.73 & 15.0 & 0.15 & 0.22 & 1.32 & 15.0 & 0.14 & 0.16 & 1.01 \\
OLS top-15      & 15.0 & 0.16 & 0.36 & 2.03 & 15.0 & 0.19 & 0.80 & 5.49 & 15.0 & 0.15 & 0.23 & 1.38 \\
\bottomrule
\end{tabular}
\end{table}

Among the methods without a hard top-$k$ truncation, LassoCV attains the lowest or essentially tied Type I and end-to-end LOSS in every cell, matched only by ElasticNetCV. We therefore use LassoCV as a standard and competitive one-shot comparator.

\subsection{Stability selection}
\label{app:ss}

Stability selection \citep{meinshausen2010stability} retains variables
that appear in at least a fraction $\pi_{\mathrm{thr}}$ of $B$ subsampled
LassoCV fits. Table \ref{tab:ss} compares it with $\Sthree$ at
$(c,q)=(0.50,0.65)$, $p=80$, $k_{\mathrm{true}}=5$ and $R=500$.
The one-shot Welch design uses the same number of runs as Stage 1 of $\Sthree$. One-shot LassoCV and stability selection use the same one-shot data. Each comparator and $\Sthree$ then use
$k_{\mathrm{m}}=\min(k^{\Sthree}_{\mathrm{sel}},k^{\mathrm{alt}}_{\mathrm{sel}},15)$ factors for the common Stage 2 protocol.

\begin{table}[H]
\centering
\caption{Stability selection ($B=100$, $\pi_{\mathrm{thr}}=0.6$) at
$(c,q)=(0.50,0.65)$, $p=80$, $k_{\mathrm{true}}=5$ and $R=500$.
Selection metrics and matched sizes are averaged over all
replications, including empty selections. The Pairs column gives the number of completed Stage 2 comparisons.
Paired $\Delta$LOSS is $\Sthree$ minus the comparator, with 95\% paired
$t$ intervals in brackets. Negative values favor $\Sthree$, and
intervals excluding zero are marked $\dagger$.}
\label{tab:ss}
\small
\setlength{\tabcolsep}{3.3pt}
\begin{tabular}{l rrrr r r r}
\toprule
Method & $k_{\mathrm{sel}}$ & T1 & T2 & F1 & matched-$k$ & Pairs
 & paired $\Delta$LOSS [95\% CI] \\
\midrule
\multicolumn{8}{l}{\textbf{Case 1 (main-effect dominated)}}\\
$\Sthree$ & 5.6 & 0.019 & 0.158 & 0.81 & --- & --- & --- (reference) \\
One-shot LassoCV & 19.5 & 0.195 & 0.026 & 0.45 & 5.55 & 500 & $-0.16\,[-0.28,-0.04]^{\dagger}$ \\
Stability selection & 4.6 & 0.009 & 0.209 & 0.82 & 4.30 & 500 & $-0.62\,[-0.75,-0.49]^{\dagger}$ \\
\midrule
\multicolumn{8}{l}{\textbf{Case 2 (weak signal)}}\\
$\Sthree$ & 5.7 & 0.039 & 0.438 & 0.54 & --- & --- & --- (reference) \\
One-shot LassoCV & 16.7 & 0.174 & 0.271 & 0.38 & 5.53 & 499 & $-0.10\,[-0.24,+0.04]$ \\
Stability selection & 2.3 & 0.005 & 0.612 & 0.51 & 2.32 & 402 & $-0.53\,[-0.68,-0.37]^{\dagger}$ \\
\bottomrule
\end{tabular}
\end{table}

At $\pi_{\mathrm{thr}}=0.6$, stability selection retains fewer than
$k_{\mathrm{true}}=5$ factors on average and misses about $21\%$ of the
active factors in Case 1 and $61\%$ in Case 2. Among completed matched-size
comparisons, the paired $\Delta$LOSS is $-0.62$ in Case 1 and $-0.53$ in
Case 2 in favor of $\Sthree$.

\subsection{Strict-Penalty Lasso baselines}
\label{app:lambda_strict}

The cross-validated penalty $\lambda_{\min}$ used throughout
(the five-fold \texttt{LassoCV} minimizer) is tuned for prediction rather than support recovery. Lasso support recovery requires additional design and signal conditions, such as the irrepresentable condition \citep{zhao2006model}. Under the default tuning used here, it therefore returns a comparatively large support. We compare $\lambda_{\min}$ against three stricter rules on the identical $\Var(s_+^{\mathrm{W}})$ designs, paired per replicate at $p = 80$, $k_{\mathrm{true}} = 5$, $R = 500$: the one-standard-error rule $\lambda_{1\mathrm{se}}$ (the largest penalty whose cross-validated error is within one standard error of the minimum), the Bayesian information criterion (BIC, minimizing $n \log(\mathrm{RSS}/n) + k \log n$, with RSS the residual sum of squares, over the models on the Lasso solution path rather than by cross-validation), and the extended BIC (EBIC), which adds a model-space penalty $2 \gamma k \log p$ with $\gamma = 1$ \citep{chen2008extended}. Table \ref{tab:lambda_strict} reports the raw selected size and the two error rates.

\begin{table}[H]
\centering
\caption{LassoCV versus strict-penalty Lasso on
$\Var(s_+^{\mathrm{W}})$ designs, $p = 80$, $k_{\mathrm{true}} = 5$,
$R = 500$, at three sample-size ratios $n/p$. Entries are the cell-mean
raw selected size $k_{\mathrm{raw}}$, the Type I (T1) and Type II (T2)
error rates, and the end-to-end natural-$k$ LOSS (capped at the Stage 2 factor limit of 15).}
\label{tab:lambda_strict}
\footnotesize
\setlength{\tabcolsep}{3.0pt}
\begin{tabular}{l rrrr rrrr rrrr}
\toprule
 & \multicolumn{4}{c}{$n/p = 0.625$}
 & \multicolumn{4}{c}{$n/p = 1.0$}
 & \multicolumn{4}{c}{$n/p = 1.5$} \\
\cmidrule(lr){2-5}\cmidrule(lr){6-9}\cmidrule(lr){10-13}
Penalty & $k_{\mathrm{raw}}$ & T1 & T2 & LOSS
        & $k_{\mathrm{raw}}$ & T1 & T2 & LOSS
        & $k_{\mathrm{raw}}$ & T1 & T2 & LOSS \\
\midrule
\multicolumn{13}{l}{\textbf{Case 1 (main-effect dominated)}}\\
$\lambda_{\min}$ (LassoCV) & 19.7 & 0.20 & 0.04 & 0.38 & 17.5 & 0.17 & 0.01 & 0.16 & 17.9 & 0.17 & 0.00 & 0.06 \\
$\lambda_{1\mathrm{se}}$   & 10.2 & 0.07 & 0.08 & 0.56 &  8.9 & 0.05 & 0.04 & 0.31 &  8.0 & 0.04 & 0.01 & 0.11 \\
BIC                        & 50.5 & 0.61 & 0.02 & 0.62 & 51.0 & 0.61 & 0.02 & 0.97 &  8.2 & 0.04 & 0.01 & 0.10 \\
EBIC ($\gamma = 1$)        &  3.2 & 0.00 & 0.42 & 2.19 &  4.7 & 0.01 & 0.16 & 1.09 &  5.3 & 0.01 & 0.05 & 0.30 \\
\midrule
\multicolumn{13}{l}{\textbf{Case 2 (weak signal)}}\\
$\lambda_{\min}$ (LassoCV) & 15.9 & 0.16 & 0.26 & 1.68 & 14.3 & 0.14 & 0.20 & 1.31 & 15.4 & 0.15 & 0.13 & 0.93 \\
$\lambda_{1\mathrm{se}}$   &  6.4 & 0.05 & 0.44 & 2.12 &  5.8 & 0.03 & 0.36 & 1.85 &  5.5 & 0.03 & 0.29 & 1.48 \\
BIC                        & 51.0 & 0.62 & 0.11 & 2.02 & 61.5 & 0.76 & 0.08 & 3.13 &  5.7 & 0.03 & 0.26 & 1.36 \\
EBIC ($\gamma = 1$)        &  0.6 & 0.00 & 0.89 & 2.71 &  1.3 & 0.00 & 0.75 & 2.41 &  2.6 & 0.00 & 0.49 & 2.07 \\
\bottomrule
\end{tabular}
\end{table}

Tightening the penalty from $\lambda_{\min}$ through $\lambda_{1\mathrm{se}}$ to EBIC drives Type I toward zero but raises Type II and end-to-end LOSS, most severely in the weak-signal Case 2, where EBIC's mean selected size falls below $2$ at $n \le p$ and its natural-$k$ Stage 2 reach drops to $0.17$--$0.45$. BIC is unusable at $n \le p$, where a near-zero residual sum of squares drives it to a near-saturated model ($k_{\mathrm{raw}} \approx 50$--$61$, Type I $\approx 0.6$--$0.76$), and recovers a sensible selection only once $n > p$. We therefore use LassoCV at $\lambda_{\min}$ as the default
baseline for comparing selected set sizes.

\section{Scorer Comparison}
\label{app:scorer}

The GQ framework is scorer-agnostic (Remark \ref{rem:scorer}). The
importance scorer is a plug-in, and the default
$|\hat\beta_j|$ can be swapped without touching the rest of the
pipeline. This appendix demonstrates the swap on the
interaction-dominated Case 3, where the default linear model is
known to fail. Case 3 has very weak main effects and amplified
second-order structure (coefficients in Table \ref{tab:scorer_case3}).
The linear model \eqref{eq:surrogate} cannot rank factors that
appear only inside products, and a $|\hat\beta_j|$-driven $\Sthree$
run terminates Stage 1 prematurely. We replace the Stage 1
importance scorer by each of four alternatives in turn,
holding
the rest of the pipeline (the $\Var(s_+^{\mathrm{W}})$ design,
the GQ schedule, the Stage 2 D-optimal
augmentation, and the BIC selection) fixed, and report the four
runs at $(c, q) = (0.50, 0.60)$, $p = 80$, $R = 500$
(Table \ref{tab:scorer_case3}).

\begin{table}[H]
\centering
\caption{Four scorers in the $\Sthree$ pipeline at $(c,q)=(0.50,0.60)$
on the interaction-dominated Case 3, with $p=80$, $k_{\mathrm{true}}=5$
and $R=500$. Main effects are $\beta_j=0.3$ for $0\le j\le4$,
with $\beta_{01}=5$, $\beta_{12}=-4$, $\beta_{00}=3$ and $\beta_{33}=-2.5$.
The scorers are distance correlation \citep{szekely2007measuring,li2012feature},
Ball correlation \citep{pan2020ball}, random forest SHAP importance
\citep{lundberg2017unified} and Gini importance \citep{breiman2001random}.
The size column gives mean numbers of selected factors before and after
the Stage 2 cap of 15. Type I and Type II use the raw selections. LOSS is evaluated after Stage 2 optimization using the capped selections. $N_{\mathrm{S1}}$ and
$N_{\mathrm{tot}}$ are mean run counts.}
\label{tab:scorer_case3}
\small
\setlength{\tabcolsep}{3pt}
\begin{tabularx}{\textwidth}{@{}l X r r r r r r@{}}
\toprule
Scorer & basis & Selected size
 & $N_{\mathrm{S1}}$ & $N_{\mathrm{tot}}$ & LOSS & Type I & Type II \\
\midrule
dCor & marginal dep.\ (distance) & 15.66 / 11.18 & 78.5 & 114.2 & 9.27 & 0.1854 & 0.649 \\
BCor & marginal dep.\ (ball) & 18.18 / 11.23 & 77.1 & 115.7 & 7.43 & 0.2102 & 0.517 \\
SHAP importance & random forest (Shapley) & 9.26 / 8.94 & 81.5 & 99.9 & 9.01 & 0.1049 & 0.722 \\
\textbf{Gini importance} & random forest (impurity) & \textbf{9.96 / 9.34} & \textbf{81.3} & \textbf{102.2} & \textbf{6.74} & \textbf{0.1063} & \textbf{0.603} \\
\bottomrule
\end{tabularx}
\end{table}

Replacing the Stage 1 scorer
moves screening quality substantially, with Gini importance
attaining the lowest LOSS ($6.74$) and a near-lowest Type I rate
($0.106$) jointly, BCor attaining the lowest Type II ($0.517$).

\section{Design Criterion Ablation}
\label{app:criterion}

This appendix justifies the choice of $\Var(s_+^{\mathrm{W}})$ as the $\Sthree$ design criterion against three alternative two-level SSD criteria: $E(s^2)$ \citep{booth1962some}, $\Var(s_+)$
\citep{weese2021strategies}, and SS\_EC \citep{singh2023selection}. We first run two ablations on the two-level design space at $p \in \{40, 80, 120, 200\}$, $R = 500$ replications. We then record the scope of the criterion for quadratics and known interactions through the $\varphi^+$ pilot.

\paragraph{Coordinate-exchange implementation.}
For each round we initialize the $n_t \times |\cO_t|$ open-factor
block with i.i.d.\ Rademacher entries, pass
through each exchangeable design entry in raster order, and
accept a sign flip if and only if it strictly decreases
$\Var(s_+^{\mathrm{W}})$ evaluated on the cumulative Gram matrix
(ties at zero improvement keep the current sign). The inner loop
iterates full sweeps until no flip is accepted in an entire sweep
or a cap of $30$ sweeps per round is reached. Fixed-set
columns are held at their committed levels and not exchanged.

\subsection{Criterion comparison}
\label{app:abl}

\textbf{Ablation A} (Table \ref{tab:abl_a}) swaps the design
criterion inside the full $\Sthree$ pipeline. Each arm constructs
its own coordinate-exchange two-level design under the indicated
criterion, then runs the same sequential GQ elimination with the
default $|\hat\beta_j|$ scoring. The result isolates the effect of
the design
criterion on end-to-end performance. Rows average over the retained
Ablation A summary cells satisfying the condition
$N_{\mathrm{S1}} < p$ for all reported arms.

\textbf{Ablation B} (Table \ref{tab:abl_b}) fixes Stage 1 at a
single LassoCV fit on a two-level design constructed under the
indicated criterion (no sequential elimination, no GQ). This
measures the criterion's effect on a single screen with a fixed run count,
holding the analysis method constant. Rows are means across all
$N_{\mathrm{S1}} < p$ cells.

\begin{table}[H]
\centering
\caption{Ablation A: design criterion swap inside the $\Sthree$
pipeline. Cell means are averaged over retained summary cells
satisfying $N_{\mathrm{S1}} < p$ for all reported arms, $R = 500$
replications. Win \% is the fraction of those cells on which the arm
has the smallest cell-mean LOSS.}
\label{tab:abl_a}
\small
\begin{tabular}{l rrr rrr}
\toprule
 & \multicolumn{3}{c}{Case 1 (808 cells)}
 & \multicolumn{3}{c}{Case 2 (890 cells)} \\
\cmidrule(lr){2-4}\cmidrule(lr){5-7}
Criterion & LOSS & T2 & Win \%
          & LOSS & T2 & Win \% \\
\midrule
$\Var(s_+^{\mathrm{W}})$
              & \textbf{2.106} & \textbf{0.340} & \textbf{99.5}\%
              & \textbf{3.579} & \textbf{0.560} & \textbf{99.4}\% \\
$E(s^2)$
              & 2.846 & 0.414 & 0.0\%
              & 4.032 & 0.672 & 0.0\% \\
$\Var(s_+)$
              & 2.755 & 0.405 & 0.0\%
              & 3.903 & 0.662 & 0.6\% \\
SS\_EC
              & 2.996 & 0.422 & 0.5\%
              & 4.127 & 0.669 & 0.0\% \\
\bottomrule
\end{tabular}
\end{table}

$\Var(s_+^{\mathrm{W}})$ has the lowest mean LOSS and the highest
cell-level win rate against every other criterion in both cases. The
gap on raw LOSS to the runner-up $\Var(s_+)$ is $-0.65$ in Case 1
and $-0.32$ in Case 2.

\begin{table}[H]
\centering
\caption{Ablation B: design criterion swap under fixed-$N_{\mathrm{S1}}$
one-shot LassoCV (no $\Sthree$ pipeline). Cell means across
$N_{\mathrm{S1}} < p$ rows, $R = 500$. The $p = 80$ rows span the
sparsity sweep $k_{\mathrm{true}} \in \{4, \ldots, 12\}$, while the
other $p$ run at $k_{\mathrm{true}} = 5$ only.}
\label{tab:abl_b}
\small
\begin{tabular}{c c rrrr}
\toprule
$p$ & rows & $\Var(s_+^{\mathrm{W}})$
          & $E(s^2)$ & $\Var(s_+)$ & SS\_EC \\
\midrule
\multicolumn{6}{l}{\textbf{Case 1 mean LOSS}}\\
40  & 4  & 1.851 & 1.951 & \textbf{1.839} & 2.325 \\
80  & 36 & \textbf{1.875} & 2.062 & 1.993 & 2.348 \\
120 & 4  & \textbf{0.295} & 0.315 & 0.327 & 0.433 \\
200 & 4  & \textbf{0.097} & 0.119 & 0.110 & 0.142 \\
\multicolumn{6}{l}{\textbf{Case 2 mean LOSS}}\\
40  & 4  & 3.386 & 3.417 & \textbf{3.273} & 3.629 \\
80  & 36 & 4.112 & 4.137 & \textbf{4.026} & 4.333 \\
120 & 4  & 1.950 & 2.002 & \textbf{1.940} & 2.120 \\
200 & 4  & \textbf{1.425} & 1.518 & 1.526 & 1.600 \\
\bottomrule
\end{tabular}
\end{table}

Two observations follow. First, in Ablation A $\Var(s_+^{\mathrm{W}})$
is unambiguously dominant. Second, in Ablation B, $\Var(s_+^{\mathrm{W}})$ leads Case 1 overall and at
$p \in \{80, 120, 200\}$, with a near tie at $p = 40$. It also leads
Case 2 at $p = 200$, while $\Var(s_+)$ leads Case 2 at
$p \in \{40, 80, 120\}$.
An arm is retained as a one-shot comparator in the main simulation
when it wins at least $3$ of the $48$ Ablation B cells. In Case 1 all
three designs qualify: l\_welch ($40/48$), l\_varsplus ($4/48$), and
l\_es2 ($4/48$). In Case 2 only l\_welch ($10/48$) and l\_varsplus
($36/48$) do.

\subsection{Criterion Scope for Quadratics and Known Interactions}
\label{app:phiplus}

The $\Var(s_+^{\mathrm{W}})$ criterion controls inner products between main-effect columns and can be extended to include known active interactions. On
the $\pm 1$ Stage 1 design a pure quadratic ($x_j^2 \equiv 1$) is
aliased with the intercept and deferred to the Stage 2
$\{-1, 0, +1\}$ pool, but a two-factor interaction is a genuine
$\pm 1$ column and can be controlled directly. When a set $\cA_{\mathrm{int}}$ of two-factor interactions
is known, the model matrix
$M = [\mathbf{x}_1, \ldots, \mathbf{x}_p]$ is augmented with the
interaction columns $\mathbf{z}_{ab} = \mathbf{x}_a \odot \mathbf{x}_b$,
$(a, b) \in \cA_{\mathrm{int}}$, and the criterion is evaluated over the
(main, main) and (main, interaction) off-diagonals of
$[\,M \mid Z\,]^{\!\top}[\,M \mid Z\,]$:
\begin{equation}
  \varphi^+_{\mathrm{welch}}(X)
  = \Var\!\bigl(s^{MM} \cup s^{MZ}\bigr)
  + \eta_{\mathrm{W}}\,\overline{(s^{MM} \cup s^{MZ})}^{\,2}
  + \lambda \!\!\sum_{s < 0}\! s^2 .
  \label{eq:phiplus}
\end{equation}
Here $\eta_{\mathrm{W}} = \eta_{\mathrm{W}}(n, p)$ is evaluated at the main-effect dimensions. This serves a similar purpose to retaining the linear main effects during BIC selection in Stage 2.

For the two active interactions of the test Cases,
$\cA_{\mathrm{int}} = \{(0,1),(1,2)\}$, we generate paired
plain-$\Var(s_+^{\mathrm{W}})$ and $\varphi^+$ designs by coordinate
exchange from a common random start ($30$ replications per $(n, p)$
cell) and measure the (main, interaction) inner products $s^{MZ}$.
Table \ref{tab:phiplus_pilot} shows that $\varphi^+$ approximately
halves $\Var(s^{MZ})$ and drives the negative-pair fraction from
$\approx0.45$ to $\approx0$, while $\max|s^{MM}|$ remains
essentially unchanged. 

\begin{table}[H]
\centering
\caption{Design-quality pilot for the $\varphi^+$ extension:
plain-$\Var(s_+^{\mathrm{W}})$ versus $\varphi^+$ with
$\cA_{\mathrm{int}} = \{(0,1),(1,2)\}$, $30$ paired coordinate-exchange
replications per $(n, p)$ cell. Columns report the variance and the
negative fraction of the (main, interaction) inner products $s^{MZ}$,
and the maximum absolute (main, main) inner product $|s^{MM}|$.}
\label{tab:phiplus_pilot}
\small
\begin{tabular}{cc rr rr rr}
\toprule
 & & \multicolumn{2}{c}{$\Var(s^{MZ})$}
   & \multicolumn{2}{c}{$\Pr(s^{MZ} < 0)$}
   & \multicolumn{2}{c}{$\max|s^{MM}|$} \\
\cmidrule(lr){3-4}\cmidrule(lr){5-6}\cmidrule(lr){7-8}
$n$ & $p$ & plain & $\varphi^+$ & plain & $\varphi^+$ & plain & $\varphi^+$ \\
\midrule
40 & 40  & 35.7 & 18.1 & 0.42 & 0.01 & 21.6 & 22.9 \\
50 & 80  & 45.4 & 20.6 & 0.46 & 0.00 & 24.1 & 24.7 \\
80 & 80  & 72.9 & 34.2 & 0.42 & 0.01 & 34.5 & 35.7 \\
60 & 120 & 52.4 & 25.6 & 0.44 & 0.00 & 30.3 & 31.1 \\
\bottomrule
\end{tabular}
\end{table}

\section{\texorpdfstring{$\Sthree$}{S3} versus One-Shot LassoCV: Cell-Level Win Rates}
\label{app:matched}

The main-text Table \ref{tab:main_results} reports aggregates at
the recommended operating point of
Section \ref{sec:main_comparison} for the l\_welch LassoCV baseline only and omits the cell-level win rate.
Table \ref{tab:band_comparison} provides the recommended-band
disaggregation, including the cell-level win rate against all three
LassoCV baselines.

\begin{table}[H]
\centering
\caption{Cell-level breakdown at the recommended operating point:
paired $\Delta\mathrm{LOSS}$ (cell-mean over recommended $c$-band $\times$
$q \in [0.40, 0.80]$ $\times$ $k = 5$) against three one-shot LassoCV baselines,
with 95\% paired-replication confidence interval half-widths ($\pm$).}
\label{tab:band_comparison}
\small
\begin{tabular}{c c c rrr c}
\toprule
 & & rec.\ $c$ & \multicolumn{3}{c}{$\Delta$LOSS vs (95\% CI half-width $\pm$)} & cells \\
\cmidrule(lr){4-6}
Case & $p$ & band
  & l\_welch & l\_varsplus & l\_es2 & won \\
\midrule
1 & 40  & $[0.55, 0.65]$ & $-0.104{\,\pm\,}0.026$ & $-0.174{\,\pm\,}0.026$ & $-0.226{\,\pm\,}0.027$ & 22/27 \\
1 & 80  & $[0.45, 0.65]$ & $-0.099{\,\pm\,}0.016$ & $-0.179{\,\pm\,}0.016$ & $-0.212{\,\pm\,}0.016$ & 39/45 \\
1 & 120 & $[0.35, 0.50]$ & $-0.054{\,\pm\,}0.017$ & $-0.154{\,\pm\,}0.018$ & $-0.191{\,\pm\,}0.018$ & 27/36 \\
1 & 200 & $[0.35, 0.50]$ & $-0.066{\,\pm\,}0.012$ & $-0.110{\,\pm\,}0.012$ & $-0.114{\,\pm\,}0.012$ & 34/36 \\
\midrule
2 & 40  & $[0.55, 0.65]$ & $-0.065{\,\pm\,}0.035$ & $+0.032{\,\pm\,}0.036$ & -- & 19/27 \\
2 & 80  & $[0.45, 0.65]$ & $-0.078{\,\pm\,}0.020$ & $-0.064{\,\pm\,}0.019$ & -- & 38/45 \\
2 & 120 & $[0.35, 0.50]$ & $-0.066{\,\pm\,}0.023$ & $-0.040{\,\pm\,}0.022$ & -- & 28/36 \\
2 & 200 & $[0.35, 0.50]$ & $+0.006{\,\pm\,}0.017$ & $-0.033{\,\pm\,}0.016$ & -- & 19/36 \\
\bottomrule
\end{tabular}
\end{table}

\section{Experiment under a Fixed Total Budget}
\label{app:fixed_budget}

This appendix reports a complementary check under a single fixed total
budget, in contrast to the matched Stage 1 budget of
Section \ref{sec:main_comparison}. The check fixes the total budget in
advance, a protocol distinct from the adaptive main method, so that
each comparator may pick a feasible model along its own path. We run it
at the tight level, the
regime the method targets, where the total budget is just enough to
screen and then fit a second-order surface. Every arm receives the same
total $N$ ($50$, $70$, $80$, and $110$ at $p = 40, 80, 120, 200$),
allocates it between screening and the surface, and may not exceed it.

$\Sthree$ runs budget-aware, the same pipeline rather than a separate
method. It reserves $p_{\mathrm{poly}}(2) + 5 = 11$ runs for the surface
before each round, stops Stage 1 once fewer than two new runs
can be added, and truncates the selected set by the same feasibility
rule the one-shot arms use, with cost ratio $c$ set to the value of
Section \ref{sec:main_comparison} for each $p$
($0.55$, $0.45$, $0.35$, $0.35$). The one-shot arms split $N$ as
$N_{\mathrm{S1}}/N\in\{0.4,0.5,0.6\}$.
We also compare $\Sthree$ against each baseline at the split
giving that baseline its lowest mean LOSS.
The remaining budget must cover $p_{\mathrm{poly}}+5$ runs,
which determines the largest feasible model size, capped at
15 factors.

The two comparators are budget-aware LassoCV, which takes the
LassoCV path point with the lowest cross-validation error among
those satisfying the model size constraint, and GDS-ARM
(adapted), a recent SSD analysis comparator \citep{singh2024factor}.
A marginal prefilter keeps GDS-ARM tractable at $p \ge 80$, and a
$p = 40$ anchor against the full method confirms the adaptation is
faithful.

We report the comparison at the default budget split
(Table \ref{tab:fb_arms}) and across budget splits
(Table \ref{tab:fb_split}).

\begin{table}[H]
\centering
\caption{Tight-budget mean LOSS and $F_1$ by arm at the default split
$N_{\mathrm{S1}}/N = 0.5$, with $R = 500$ replications per cell. $\Sthree$ allocates its runs adaptively. BA-LassoCV is the budget-aware LassoCV and GDS-ARM is the adapted
comparator.}
\label{tab:fb_arms}
\small
\begin{tabular}{c c rr rr rr}
\toprule
 & & \multicolumn{2}{c}{$\Sthree$} & \multicolumn{2}{c}{BA-LassoCV}
   & \multicolumn{2}{c}{GDS-ARM} \\
\cmidrule(lr){3-4}\cmidrule(lr){5-6}\cmidrule(lr){7-8}
Case & $p$ & LOSS & $F_1$ & LOSS & $F_1$ & LOSS & $F_1$ \\
\midrule
1 & 40  & $1.63$ & $0.71$ & $3.29$ & $0.66$ & $3.14$ & $0.70$ \\
1 & 80  & $1.24$ & $0.77$ & $2.03$ & $0.69$ & $2.32$ & $0.75$ \\
1 & 120 & $1.12$ & $0.83$ & $1.96$ & $0.69$ & $2.15$ & $0.74$ \\
1 & 200 & $0.66$ & $0.86$ & $1.25$ & $0.66$ & $1.76$ & $0.76$ \\
\midrule
2 & 40  & $3.00$ & $0.51$ & $3.65$ & $0.44$ & $3.75$ & $0.48$ \\
2 & 80  & $2.56$ & $0.54$ & $2.92$ & $0.47$ & $3.58$ & $0.48$ \\
2 & 120 & $2.52$ & $0.53$ & $3.14$ & $0.45$ & $3.64$ & $0.47$ \\
2 & 200 & $2.09$ & $0.57$ & $2.57$ & $0.44$ & $3.20$ & $0.50$ \\
\bottomrule
\end{tabular}
\end{table}

\begin{table}[H]
\centering
\caption{Tight-budget split sensitivity, pooled over $p$.
Entries are the comparator's LOSS minus $\Sthree$'s LOSS, so positive
values favor $\Sthree$. Split labels apply to the one-shot arms.
Each split column uses the $\Sthree$ results from the corresponding
comparison. The best column takes the comparator's lowest mean LOSS
across splits at each $p$, with the $\Sthree$ reference from the
$0.5$ comparison.}
\label{tab:fb_split}
\small
\begin{tabular}{c rrrr rrrr}
\toprule
 & \multicolumn{4}{c}{budget-aware LassoCV}
   & \multicolumn{4}{c}{GDS-ARM (adapted)} \\
\cmidrule(lr){2-5}\cmidrule(lr){6-9}
Case & $0.4$ & $0.5$ & $0.6$ & best
   & $0.4$ & $0.5$ & $0.6$ & best \\
\midrule
1 & $+1.11$ & $+0.97$ & $+0.55$ & $+0.61$ & $+1.86$ & $+1.18$ & $+0.60$ & $+0.66$ \\
2 & $+0.69$ & $+0.53$ & $+0.16$ & $+0.23$ & $+1.48$ & $+1.00$ & $+0.28$ & $+0.35$ \\
\bottomrule
\end{tabular}
\end{table}

Against both comparators, $\Sthree$ has the lower LOSS and the higher
$F_1$ in every one of the eight case-by-$p$ cells
(Table \ref{tab:fb_arms}), and the LOSS advantage holds at every budget
split, including each comparator's most favourable split
(Table \ref{tab:fb_split}).

\section{Borehole Factor-Level Analysis}
\label{app:borehole}

This appendix supports the factor-level reading of
Section \ref{sec:borehole}. Table \ref{tab:borehole_factor} reports the
selection frequency of each of the eight physical factors and the mean
inactive selection rate at the representative cell $(c, q) = (0.50, 0.65)$.

\begin{table}[H]
\centering
\caption{Borehole factor-level selection frequency (\%) at the
representative cell $(c, q) = (0.50, 0.65)$, $R = 500$ replications. The bottom row is
the mean selection frequency across the $72$ inactive factors
($x_8, \ldots, x_{79}$).}
\label{tab:borehole_factor}
\small
\begin{tabular}{l l rrrr}
\toprule
factor & physical name & $\Sthree$ & l\_welch & l\_varsplus & l\_es2 \\
\midrule
$x_0$ & $r_w$ (borehole radius)            & $100.0$ & $100.0$ & $100.0$ & $100.0$ \\
$x_1$ & $r$ (radius of influence)          & $1.4$   & $22.4$  & $18.0$  & $15.4$  \\
$x_2$ & $T_u$ (upper transmissivity)       & $1.6$   & $22.6$  & $20.4$  & $17.2$  \\
$x_3$ & $H_u$ (upper head)                 & $93.4$  & $98.4$  & $99.4$  & $99.6$  \\
$x_4$ & $T_l$ (lower transmissivity)       & $1.4$   & $20.8$  & $20.8$  & $19.0$  \\
$x_5$ & $H_l$ (lower head)                 & $89.8$  & $98.6$  & $99.4$  & $99.4$  \\
$x_6$ & $L$ (borehole length)              & $91.4$  & $98.4$  & $98.6$  & $99.2$  \\
$x_7$ & $K_w$ (hydraulic conductivity)     & $43.4$  & $73.8$  & $71.2$  & $70.2$  \\
\midrule
\multicolumn{2}{l}{inactive mean ($x_8$--$x_{79}$, $72$ factors)} & $1.5$ & $21.0$ & $18.9$ & $18.4$ \\
\bottomrule
\end{tabular}
\end{table}

Figure \ref{fig:borehole_heatmap} shows the cell-mean LOSS gaps
against the lowest LassoCV mean LOSS at each $(c,q)$ setting.
$\Sthree$ wins $57/63$ cells with a mean gap of $-2.25$.
The six exceptions lie in the small-$c$, small-$q$ corner.

\begin{figure}[!htbp]
\centering
\includegraphics[width=0.70\textwidth]{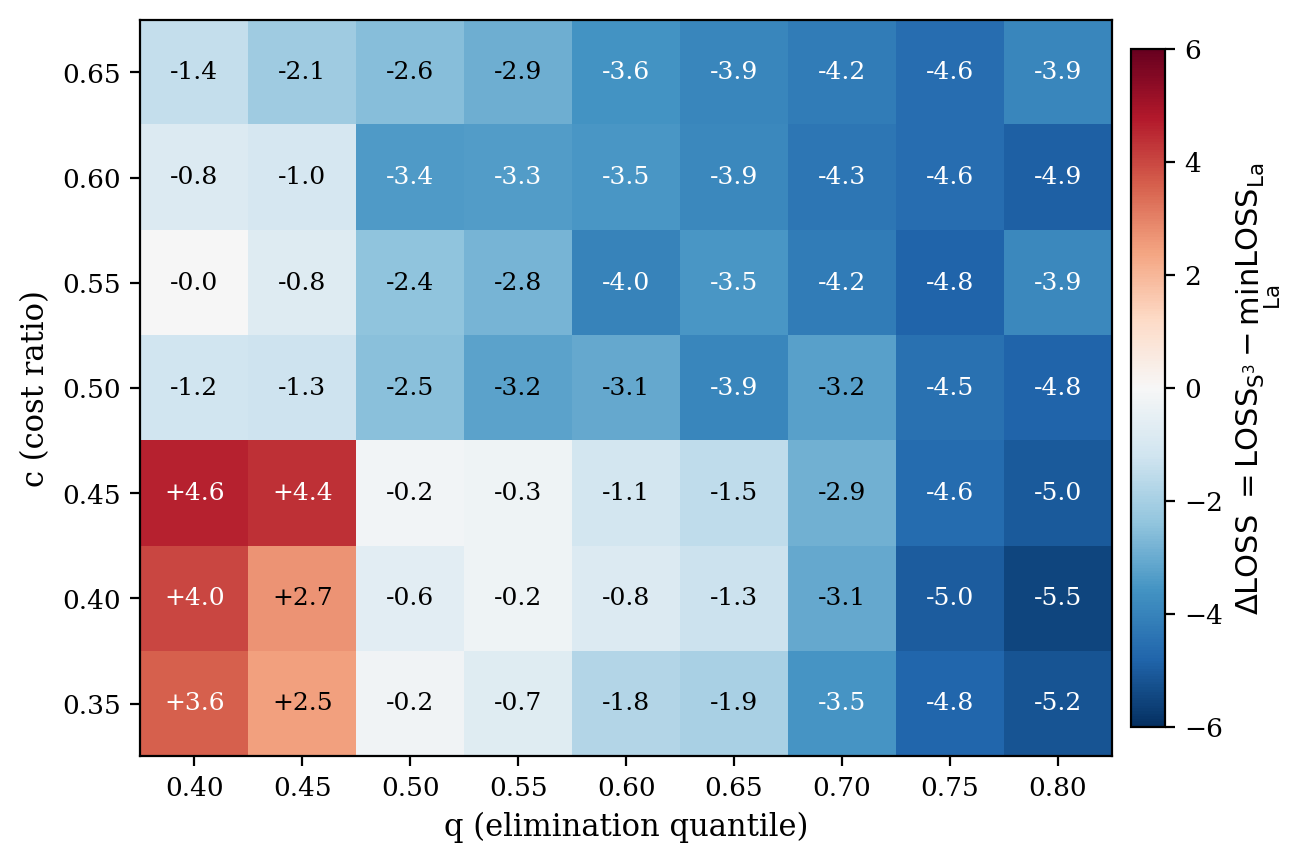}
\caption{Borehole cell-mean LOSS gap
$\overline{\mathrm{LOSS}}_{\Sthree}
-\min_{\mathrm{La}}\overline{\mathrm{LOSS}}_{\mathrm{La}}$
over the $7\times9$ $(c,q)$ grid.
The minimum is taken across the cell means for
l\_welch, l\_varsplus and l\_es2.
Negative values favor $\Sthree$ and positive values favor one-shot LassoCV.}
\label{fig:borehole_heatmap}
\end{figure}

\section{\texorpdfstring{$\Sthree$}{S3} versus a Naive Sequential Halving Baseline}
\label{app:nsh}

This appendix isolates the joint contribution of the four $\Sthree$
innovations the framework wraps around the LassoCV model fit:
the $\Var(s_+^{\mathrm{W}})$ coordinate-exchange design, the
GQ schedule, the elbow-fixing rule, and the sign commitment. We compare
$\Sthree$ against a Naive Sequential Halving (NSH) baseline that
strips all four simultaneously. At each round, NSH
allocates $n_t = \max(2, \lceil c|\cC_t|\rceil)$ runs from a
random two-level design, fits LassoCV with round-block
indicators, retains the top
$\max(k^\star, \lfloor |\cC_t| / 2 \rfloor)$ candidates by
$|\hat\beta_j|$ (with target $k^\star = 5$), and eliminates the
remainder. There is no GQ, no elbow fixing, no
sign commitment, and no design optimization. The Stage 2
response surface pipeline is identical to $\Sthree$.

We use $\Sthree$'s realized $N_{\mathrm{S1}}$ as the Stage 1
run target for NSH. The realized counts differ by at most one run.
We compare both arms at the matched count
$k_{\mathrm{m}}=\min(k^{\Sthree}_{\mathrm{sel}},
k^{\mathrm{NSH}}_{\mathrm{sel}})$, so the Stage 2 augmentation and total
budgets are closely matched. Table \ref{tab:nsh_results} aggregates
$R=500$ replications at the representative cell
$(c,q)=(0.50,0.65)$ across $p\in\{40,80,120,200\}$ and the two cases.

\begin{table}[H]
\centering
\caption{$\Sthree$ versus NSH at representative cell $(c, q) = (0.50, 0.65)$ with
matched Stage 1 budget, matched-$k$ comparison. $\Delta_{\mathrm{m}} = \mathrm{LOSS}^{\Sthree} -
\mathrm{LOSS}^{\mathrm{NSH}} < 0$ indicates $\Sthree$ wins.}
\label{tab:nsh_results}
\small
\begin{tabular}{cc r rr r}
\toprule
Case & $p$
  & $k_{\mathrm{m}}$
  & $\mathrm{LOSS}^{\Sthree}$ & $\mathrm{LOSS}^{\mathrm{NSH}}$
  & $\Delta_{\mathrm{m}}$ \\
\midrule
1 & 40  & 5.00 & 1.909 & 4.093 & $-2.184$ \\
1 & 80  & 5.59 & 1.085 & 1.563 & $-0.478$ \\
1 & 120 & 5.81 & 0.684 & 1.065 & $-0.381$ \\
1 & 200 & 5.76 & 0.426 & 0.642 & $-0.216$ \\
2 & 40  & 4.67 & 3.316 & 4.848 & $-1.533$ \\
2 & 80  & 5.99 & 2.427 & 3.025 & $-0.598$ \\
2 & 120 & 5.79 & 2.101 & 2.531 & $-0.430$ \\
2 & 200 & 6.18 & 1.716 & 2.040 & $-0.324$ \\
\bottomrule
\end{tabular}
\end{table}

$\Sthree$ wins the matched-$k$ comparison in every cell. The mean
$\Delta_{\mathrm{m}}$ is $-0.815$ for Case 1 and $-0.721$ for Case 2.
The largest gap is $-2.184$ at $p = 40$ in Case 1, and the magnitude
decreases as $p$ grows in both cases.

\section{Sensitivity to the Positive-Cone Penalty Weight \texorpdfstring{$\lambda$}{lambda}}
\label{app:lambda_sweep}

We use $\lambda=10^6$ as the default positive-cone penalty
weight in \eqref{eq:welch_crit}.
We examine sensitivity to this choice using the full two-stage
procedure at the representative cell with five penalty weights.

\begin{table}[H]
\centering
\caption{End-to-end $\Sthree$ output for five penalty weights in
Case 1 at $p=80$, $k_{\mathrm{true}}=5$,
$(c,q)=(0.50,0.65)$ and $R=500$.
$k_{\mathrm{sel}}$ is the raw selected size.
Type I, Type II, $F_1$ and MCC use selections capped at 15 factors.
The final column reports the proportion of selected Stage 1
submatrices, before capping, with $\min_{i<j}s_{ij}\ge0$.}
\label{tab:lambda_sweep}
\small
\begin{tabular}{r r r r r r r r}
\toprule
$\lambda$ & LOSS & TypeI & TypeII & $F_1$ & MCC & $k_{\mathrm{sel}}$ & Positive cone \\
\midrule
$10^{4}$ & $1.121$ & $0.019$ & $0.156$ & $0.805$ & $0.802$ & $5.67$ & $95.2\%$ \\
$10^{5}$ & $1.091$ & $0.017$ & $0.154$ & $0.819$ & $0.816$ & $5.57$ & $97.6\%$ \\
$10^{6}$ & $1.127$ & $0.020$ & $0.156$ & $0.804$ & $0.801$ & $5.72$ & $97.8\%$ \\
$10^{7}$ & $1.107$ & $0.020$ & $0.153$ & $0.805$ & $0.802$ & $5.72$ & $95.0\%$ \\
$10^{8}$ & $1.070$ & $0.018$ & $0.152$ & $0.816$ & $0.812$ & $5.59$ & $96.6\%$ \\
\bottomrule
\end{tabular}
\end{table}

The reported metrics vary little across the five penalty
weights at this setting. Mean LOSS ranges from $1.070$ to
$1.127$, a difference of $0.057$. The proportion of designs
in the positive cone ranges from $95.0\%$ to $97.8\%$.

\section{The Complete Two-Stage \texorpdfstring{$\Sthree$}{S3} Procedure}
\label{app:fullalg}

Algorithm \ref{alg:s3} states the full $\Sthree$ procedure.

\begin{algorithm}[H]
\caption{The complete two-stage $\Sthree$ procedure}
\label{alg:s3}
\small
\begin{algorithmic}[1]
\Require $p$ candidate factors, cost ratio $c$, elimination
         quantile $q$, maximum rounds $S_{\max}$
\Ensure Estimated optimum $\hat{\mathbf{x}}$ when Stage 2 completes
\Statex \textbf{Stage 1 (sequential screening).}
\State Initialize $\cC \leftarrow [p]$,\;
       $\cF \leftarrow \emptyset$,\; $t \leftarrow 1$
\While{$\cC \setminus \cF \ne \emptyset$}
  \State $n_t \leftarrow \max(2, \lceil c\,|\cC \setminus \cF| \rceil)$
  \State Generate the $n_t$-run two-level design on $\cC$ via coordinate exchange
         with $\Var(s_+^{\mathrm{W}})$ on the cumulative matrix,\;
         holding $\cF$ at committed levels
  \State Collect responses, update the cumulative data $(X,\mathbf{y},\mathbf{B})$, and set $N_t\leftarrow|\mathbf{y}|$
  \State Fit LassoCV model \eqref{eq:surrogate}, compute $I_j$
         via \eqref{eq:impscore}, record nonzero coefficient signs
  \State $q_t \leftarrow q \cdot \min(1,\; N_t / |\cC|)$,\;
         $\tau_t \leftarrow$ $q_t$-quantile of
         $\{I_j : j \in \cC \setminus \cF\}$
  \State $E \leftarrow \{j \in \cC \setminus \cF : I_j \le \tau_t\}$,
         the factors marked for elimination
  \State Compute the elbow $r^\star$ on the open scores, add the
         top-$r^\star$ open factors with recorded signs to $\cF$ at
         their committed levels (a single open factor is added directly)
  \State $\cC \leftarrow \cC \setminus (E \setminus \cF)$, dropping the
         eliminated columns from $X$
  \State \textbf{if} $\cF$ did not grow and $t \ge 2$, or $t = S_{\max}$:
         exit the loop
  \State $t \leftarrow t + 1$
\EndWhile
\State $\hat\cA \leftarrow \cF$ with committed signs,\;
       $k_{\mathrm{sel}} \leftarrow |\hat\cA|$
\Statex \textbf{Stage 2 (response surface).}
\State Apply the size rules of Section \ref{sec:metrics} to form
       the Stage 2 subset $\mathcal A_2\subseteq\hat\cA$
\State \textbf{if} $|\mathcal A_2|<2$: \Return no Stage 2 estimate
\State Augment the cumulative design on $\mathcal A_2$ using
       the allocation and sampled greedy forward rule
\State Fit by OLS with BIC forward selection, keeping all linear main effects and block terms
\State Optimize the fitted polynomial without block terms on
       $[-1,1]^{|\mathcal A_2|}$, giving
       $\hat{\mathbf{x}}_{\mathcal A_2}$
\State Set $\hat{\mathbf{x}}_{\cF\setminus\mathcal A_2}$ to committed levels
       and $\hat{\mathbf{x}}_{[p]\setminus\cF}$ to zero
\State \Return $\hat{\mathbf{x}}$
\end{algorithmic}
\end{algorithm}

\end{document}